\documentclass[10pt]{article}

\usepackage{amsfonts}
\usepackage{amsmath}
\usepackage{amssymb}
\usepackage{amsthm}
\usepackage{enumerate}
\usepackage{mdwlist}
\usepackage{mathrsfs}
\usepackage{bussproofs}
\usepackage{graphicx}
\usepackage{latexsym}
\usepackage{tcolorbox}
\usepackage{titlesec}
\usepackage{stmaryrd}
\usepackage{verbatim}
\usepackage{url}

\titleformat*{\section}{\large\bfseries}
\titleformat*{\subsection}{\normalsize\bfseries}
\titleformat*{\subsubsection}{\normalsize\em}

\newtheorem{theorem}{Theorem}[section]

\newtheorem{lemma}[theorem]{Lemma}
\newtheorem{corollary}[theorem]{Corollary}

\theoremstyle{remark}

\newtheorem{definition}[theorem]{Definition}

\newtheorem{example}[theorem]{Example}

\newcommand{\NN}{\mathbb{N}}
\newcommand{\BB}{\mathbb{B}}
\newcommand{\nat}{\mathsf{N}}
\newcommand{\bool}{\mathsf{B}}

\newcommand{\comp}[2]{[{#1}]_{#2}}
\newcommand{\ZL}{\mathrm{ZL}}
\newcommand{\LEX}{\mathrm{LEX}}

\newcommand{\modset}{\mathcal{S}^\omega}
\newcommand{\modcont}{\mathcal{C}^\omega}
\newcommand{\modpar}{{\mathcal{P}}^\omega}

\newcommand{\cont}[3]{\{#1\}^{#2}_{#3}}
\newcommand{\conte}[2]{\langle{#1}\rangle_{#2}}
\newcommand{\contd}[3]{\langle{#1}\rangle^{#2}_{#3}}

\newcommand{\dt}[3]{|{#1}|^{#2}_{#3}}
\newcommand{\HAw}{\mathrm{HA}^\omega}
\newcommand{\PAw}{\mathrm{PA}^\omega}

\newcommand{\EPAw}{\mathrm{E\mbox{-}PA}^\omega}
\newcommand{\WHAw}{\mathrm{WE\mbox{-}HA}^\omega}
\newcommand{\WPAw}{\mathrm{WE\mbox{-}PA}^\omega}

\newcommand{\tail}{\mathrm{tail}}
\newcommand{\spec}{\eta}
\newcommand{\SR}{\mathrm{RP}}
\newcommand{\ub}[1]{\tilde{#1}}
\newcommand{\seq}[1]{[ {#1} ]}
\newcommand{\types}{\bf{T}}
\newcommand{\bneg}[1]{\overline{#1}}
\newcommand{\num}[1]{\underline{#1}}
\newcommand{\QFAC}{\mathrm{QF\mbox{-}AC}}
\newcommand{\at}{\; @ \;}
\newcommand{\exts}[2]{\overline{{#1},{#2}}}

\title{On the computational content of Zorn's lemma}
\author{Thomas Powell}
\date{}

\begin{document}

\maketitle

\begin{abstract}
We give a computational interpretation to an abstract instance of Zorn's lemma formulated as a wellfoundedness principle in the language of arithmetic in all finite types. This is achieved through G\"odel's functional interpretation, and requires the introduction of a novel form of recursion over non-wellfounded partial orders whose existence in the model of total continuous functionals is proven using domain theoretic techniques. We show that a realizer for the functional interpretation of open induction over the lexicographic ordering on sequences follows as a simple application of our main results.\medskip

\end{abstract}

\section{Introduction}
\label{sec-intro}

The correspondence between proofs and programs is one of the most fundamental ideas in computer science. Initially connecting intuitionistic logic with the typed lambda calculus, it has since been extended to incorporate a wide range of theories and programming languages. 

A challenging problem in this area is to give a computational interpretation to the axiom of choice in the setting of classical logic. A number of ingenious solutions have been proposed, ranging from Spector's fundamental consistency proof of classical analysis using bar recursion \cite{Spector(1962.0)} to more modern approaches, which include the \emph{Berardi-Bezem-Coquand} functional \cite{BBC(1998.0)}, optimal strategies in sequential games \cite{EscOli(2011.0)}, and Krivine's `quote' and `clock' \cite{Krivine(2003.0)}.

In this paper, we introduce both a new form of recursion and a new computational interpretation of a choice axiom. In contrast to the aforementioned works, which all focus on variants of countable choice, we give a direct computational interpretation to an axiomatic formulation of Zorn's lemma. Our work is closest in spirit to Berger's realizability interpretation of open induction on the lexicographic ordering via open recursion \cite{Berger(2004.0)} - an idea which was later transferred to the setting of G\"{o}del's functional interpretation in \cite{Powell(2018.0)}. However, a crucial difference here is that we do not work with a concrete order, but a general parametrised variant of Zorn's lemma, from which induction on the lexicographic ordering can be considered a special case.

After formulating an axiomatic version of Zorn's lemma in the language of Peano arithmetic in all finite types, we study related forms of recursion on chain bounded partial orders. In particular, we introduce a new recursive scheme based on the notion of a `truncation', and give precise domain theoretic conditions under which the resulting fixpoint in the partial continuous functionals is total (Theorem \ref{thm-totcont}).

We then demonstrate that we can use our new form of recursion to solve the functional interpretation of our variant of Zorn's lemma. Our approach completely separates the issues of correctness (that our program does what it's supposed to do) with that of totality (that our program is well-defined). The main correctness result (Theorem \ref{thm-mainver}) is extremely general, and its proof short and direct, suggesting that our realizing terms are natural in a fundamental way. To establish totality we make use of our earlier domain theoretic results, and again provide conditions which ensure that our computational interpretation is satisfied in the continuous functionals. We conclude with a concrete example which ties everything together, demonstrating that the functional interpretation of open induction over the lexicographic ordering can be given as a special case of our general result. 

This work aims to achieve several things. Our new recursive schemes on chain bounded partial orders form a contribution to higher-order computability theory, which we believe is of interest in its own right. The subsequent computational interpretation of Zorn's lemma is a new result in proof theory, which we hope will lead to novel applications in future work. Finally, through our general and abstract setting we provide some fresh insights into known computational interpretations of variants of the axiom of choice, particularly open recursion \cite{Berger(2004.0)} and Spector's original bar recursion \cite{Spector(1962.0)}.

\section{Preliminaries}
\label{sec-prelim}

We begin by presenting some essential background material. G\"{o}del's functional interpretation, which only appears from Section \ref{sec-dial} onwards, will be introduced later.

\subsection{Zorn's lemma}
\label{sec-prelim-zorn}

Zorn's lemma is central to this article, and features not only as a proof technique but also in the guise of an axiomatic principle. In what follows, $<$ will always denote a strict partial order, and $\leq$ its reflexive closure.
\begin{definition}
\label{def-cc}
We call a partially ordered set $(S,<)$ \emph{chain bounded} if every nonempty chain $\gamma\subseteq S$ (i.e. nonempty totally ordered subset of $S$) has an upper bound in $S$, that is an element $u\in S$ such that $x\leq u$ for all $x\in\gamma$.
\end{definition}
\begin{theorem}[Zorn's lemma]
\label{thm-zorn}
Let $(S,<)$ be a nonempty partially ordered set which is chain bounded. Then $S$ contains at least one maximal element, that is an element $x\in S$ such that $\neg(x<y)$ for all $y\in S$.
\end{theorem}
The following well-known application of Zorn's lemma will form a running illustration throughout the paper:
\begin{example}
\label{ex-ideal}
Let $R$ be some nontrivial ring with unity, and define $(S,\subset)$ to be the set of all proper ideals of $R$ partially ordered by the strict subset relation. Then $S$ is nonempty since $\{0\}\in S$, and is also chain bounded since for any nonempty chain $\gamma$, the set $\bigcup_{x\in\gamma} x$ is also a proper ideal of $R$ and thus an element of $S$. Therefore by Zorn's lemma, $S$ has a maximal element, or in other words, $R$ has a maximal ideal.
\end{example}
Our ability to apply Zorn's lemma to establish the existence of maximal ideals relies crucially on the fact that the upper bound $\bigcup_{x\in\gamma} x$ is also a proper ideal. This in turn is due to the fact that $x$ being a proper ideal is a `piecewise' property, in that it can be reduced to an infinite conjunction ranging over finite pieces of information about $x$. We now make this intuition precise, leading to a modification of Zorn's lemma (Theorem \ref{thm-open}) close in spirit to \emph{open induction} as studied by Raoult \cite{Raoult(1988.0)}. This will form the basis of our syntactic version of Zorn's lemma presented in Section \ref{sec-syntactic}.


\begin{definition}
An approximation function on the set $X$ relative to some sets $D$ and $U$ is taken to be a mapping $\comp{\cdot}{(\cdot)}:X\times D\to U$, where the sets $D$ and $U$ will play the following intuitive roles:
\begin{itemize}

\item $D$ is an index set of `sizes',

\item $U$ is a set of `approximations' of elements of $X$.

\end{itemize}
We call $\comp{x}{d}\in U$ the approximation of $x$ of size $d$.
\end{definition}
%
%
\begin{definition}
\label{def-ccapp}
We say that $(X,<)$ is chain bounded with respect to the approximation function $\comp{\cdot}{}:X\times D\to U$ if any nonempty chain $\gamma\subseteq X$ has an upper bound $\ub{\gamma}\in X$ satisfying the additional property that for all $d\in D$ there is some $x\in \gamma$ such that $\comp{\ub{\gamma}}{d}=\comp{x}{d}$.
\end{definition}
\begin{example}
\label{ex-ccapp}
Let $(2^R,\subset)$ be the powerset of some set $R$, and $D$ the set of all \emph{finite} subsets of $R$. Let
\begin{equation*}
U:=\{f:d\to \{0,1\}\; | \; d\in D\}
\end{equation*}
and define $\comp{x}{d}:d\to\{0,1\}$ by
\begin{equation*}
\comp{x}{d}(a)=1\Leftrightarrow a\in x.
\end{equation*}
Then $(2^R,\subset)$ is chain bounded with respect to $\comp{\cdot}{}$. To see this, given a chain $\gamma$ let $\ub{\gamma}:=\bigcup_{x\in\gamma} x$ and suppose that $a\in \ub{\gamma}$. Then there must be some $x_a\in\gamma$ such that $a\in x_a$. For $d\in D$ define $x:=\max_{\subset}\{x_a\; | \; a\in d\cap\ub{\gamma}\}\in\gamma$, and note that $x$ is well defined since $\gamma$ is totally ordered. Now, if $\comp{\ub{\gamma}}{d}(a)=1$ then $a\in d\cap\ub{\gamma}$ and thus $a\in x$, and so $\comp{x}{d}(a)=1$. On the other hand, if $\comp{\ub{\gamma}}{d}=0$ then $a\notin \ub{\gamma}$ and so $a\notin x$ (since $a\in x$ trivially implies $a\in\ub{\gamma}$, hence $\comp{x}{d}(a)=0$. Therefore $\comp{\ub{\gamma}}{d}=\comp{x}{d}$.
\end{example}
\begin{definition}
\label{def-piece}
We call a predicate $P(x)$ on $X$ \emph{piecewise} with respect to the approximation function $\comp{\cdot}{}:X\times D\to U$ if $P(x)\Leftrightarrow (\forall d\in D) Q(\comp{x}{d})$ for some predicate $Q(u)$ on $U$.
\end{definition}
\begin{theorem}
\label{thm-open}
Let $(X,<)$ be a partially ordered set which is chain bounded w.r.t. the approximation function $\comp{\cdot}{}:X\times D\to U$, and $P(x)$ a predicate on $X$ which is piecewise w.r.t the same function. Then whenever $P(x)$ holds for some $x\in X$, there exists $y\in X$ such that $P(y)$ holds but $\neg P(z)$ whenever $y<z$.
\end{theorem}
\begin{proof}
Let $S:=\{x\in X\; | \; P(x)\}$, and take some nonempty chain $\gamma\subseteq S$. Our first step is to show that $\ub{\gamma}\in S$, from which it follows that $(S,<)$ is chain bounded. Since $P(x)\Leftrightarrow (\forall d\in D) Q(\comp{x}{d})$ for some predicate $Q(u)$, it suffices to show that $Q(\comp{\ub{\gamma}}{d})$ for all $d\in D$. But using that for any $d$ there exists some $x\in\gamma\subseteq S$ with $\comp{\ub{\gamma}}{d}=\comp{x}{d}$ we're done, since $Q(\comp{x}{d})$ follows from $P(x)$. Now, suppose that $P(x)$ holds for some $x\in X$, and thus $S$ is nonempty. By Zorn's lemma, $S$ contains a maximal element $y$. We clearly have $P(y)$, and if $y<z$ then $z\notin S$ and thus $\neg P(z)$.
\end{proof}
\begin{example}
\label{ex-open}
Let $(2^R,\subset)$ be the powerset of some nontrivial ring $R$, with $\comp{x}{d}$ defined as in Example \ref{ex-ccapp}, and let $P(x)$ be denote the predicate `$x$ is a proper ideal of $R$'. Then this is a piecewise predicate w.r.t. $\comp{\cdot}{}$, since each condition of being a proper ideal can be formulated in a piecewise way. For instance, $0\in x$ is equivalent to 
\begin{equation*}
\forall d\; (0\in d\Rightarrow \comp{x}{d}(0)=1)
\end{equation*}
and analogously for $1\notin x$. Similarly, closure of $x$ under addition can be formulated in a piecewise way as
\begin{equation*}
\begin{aligned}
\forall d,r,r'\; (\{r,r',r+r'\}\subseteq d& \wedge \comp{x}{d}(r)=\comp{x}{d}(r')=1\\
&\Rightarrow \comp{x}{d}(r+r')=1)
\end{aligned}
\end{equation*}
and analogously for closure under left and right sided multiplication. Therefore the existence of a maximal ideal also follows from Theorem \ref{thm-open} above. Note that since in $r,r'$ above are always elements of the finite set $d$, $\forall r,r'$ can be treated as a bounded quantifier, and so `$x$ is a proper ideal of $R$' is piecewise even with respect to some quantifier-free $Q(u)$.
\end{example}
%

\subsection{Formal theories of arithmetic}
\label{sec-prelim-formal}

In the remainder of this article, our definitions and results typically take place in one of the following settings:
\begin{itemize}

\item Within a formal theory of arithmetic in higher-types (syntactic);

\item Within a type structure of continuous functionals, either the total or partial (semantic).

\end{itemize}
We now outline both of these settings in turn. Our basic formal system will be the standard theories of Peano (resp. Heyting) arithmetic in all finite types $\PAw$ ($\HAw$). For us, the finite types $\types$ will be generated by the following grammar:
\begin{equation*}
\rho,\tau::=\bool\; | \; \nat \; | \; \rho\times\tau \; | \; \rho^\ast\; | \; \rho\to\tau
\end{equation*}
These represent base types for booleans $\bool$ and natural numbers $\nat$, and in addition to the usual function type $\rho\to\tau$ include cartesian products $\rho\times\tau$ and finite sequence types $\rho^\ast$ as primitives. Note that alternatively, we could work over a minimal type structure $\nat\; | \; \rho\to\tau$ and code up products and finite sequences as derived constructions.

For full definitions of $\PAw$ resp. $\HAw$ the reader is directed to e.g. \cite{AvFef(1998.0),Kohlenbach(2008.0),Troelstra(1973.0)}, bearing in mind that officially we would need to extend the canonical theories presented there with additional constants and axioms for dealing with cartesian products and list operations, which is nevertheless entirely standard (for details see e.g. \cite[Chapter I.8]{Troelstra(1973.0)} and \cite{BergBrisSaf(2012.0)}).

Terms of $\PAw$ resp. $\HAw$ are those of G\"{o}del's System T (with product and sequence types). We denote by $0_\rho:\rho$ a canonical zero object of type $\rho$. Formulas of $\PAw$ (resp. $\HAw$) include atomic formulas $=_\bool$ and $=_\nat$ for equality at base types, and are built using the usual logical connectives, together with quantifiers for each type. Axioms and rules include those of full classical (resp. intuitionistic) logic, non-logical axioms for the constants symbols together with equality axioms and the axiom of induction. Equality at higher types is defined inductively e.g. $f=_{\rho\to\tau} g:=\forall x^\rho (fx=gx)$, and we include axioms for extensionality, so that our formulation of $\PAw$ corresponds to the fully extensional $\EPAw$ of \cite{Kohlenbach(2008.0)}. 

The canonical models for $\PAw$ include the type structures of all set-theoretic functionals $\modset$ together with total continuous functional $\modcont$. However, the majority of recursive schemes which have been used to interpret the axiom of choice (including essentially all known variants of bar recursion) are no longer satisfiable in $\modset$, and instead have $\modcont$ as their canonical model. In the remainder of this section, we outline some key facts about this model.

\subsection{The continuous functionals in all finite types}
\label{sec-prelim-cont}

In one sentence, the type structure $\modcont$ of continuous functionals consists of functionals which only require a finite piece of information about their input to compute a finite piece of information about their output. Over the years, they have turned out to form an elegant and robust class of functionals, and in particular are the standard model for \emph{bar recursive} extensions of the primitive recursive functionals.

There are various ways of characterising the continuous functionals, dating back to Kleene \cite{Kleene(1959.0)} (whose construction was based on associates) and Kreisel \cite{Kreisel(1959.0)} (who instead used formal neighbourhoods). However, here we follow the domain theoretic approach of Ershov \cite{Ershov(1977.0)}, who demonstrated that the continuous functionals can be constructed as the extensional collapse of the total objects in the type structure $\modpar$ of \emph{partial} continuous functionals. This in particular provides us with a simple method for showing that our new recursive schemes are satisfied in $\modcont$, namely proving that the corresponding fixpoints in $\modpar$ represent total objects. For accomprehensive account of all this, the reader is encouraged to consult \cite{Normann(1999.0)} or the recent book \cite{LongNor(2015.0)}. Here we provide no more than a brief overview of the relevant theory.

For each finite type $\sigma$, we define the domain $P_\sigma$ of partial continuous functionals of that type as follows: $P_\bool:=\BB_\bot$ and $P_\nat:=\NN_\bot$ where $\BB_\bot$ resp. $\NN_\bot$ are the usual flat domains of booleans and natural numbers, $P_{\rho\times\tau}:=P_\rho\times P_\tau$, $P_{\sigma^\ast}:=\{\seq{x_0,\ldots,x_{n-1}}\; | \; n\in\NN\mbox{ and }x_i\in P_\sigma\}\cup\{\bot\}$ and finally $P_{\rho\to\tau}:=[P_\rho\to P_\tau]$ where $[D\to E]$ denotes the domain of all functions between $X$ and $Y$ which are continuous in the \emph{domain theoretic} sense (i.e. are monotone and preserve lubs of chains). We write $\modpar:=\{P_\sigma\}_{\sigma\in\types}$ for this type structure of partial continuous functionals.

For each type $\sigma$, we define the set $T_\sigma\subset P_\sigma$ of total objects in the usual way as $T_\bool:=\BB$ and $T_\nat:=\NN$, $T_{\rho\times\tau}:=T_\rho\times T_\tau$, $T_{\sigma^\ast}:=\{\seq{x_0,\ldots,x_{n-1}}\; | \; n\in\NN\mbox{ and }x_i\in T_\sigma\}$ and finally $T_{\rho\to\tau}:=\{f\in P_{\rho\to\tau}\; : \; \forall x(x\in T_\sigma \Rightarrow fx\in T_\tau)\}$. Furthermore, we define an equivalence relation $\approx_\sigma$ on $T_\sigma$ to equate total objects that agree on total inputs: $x\approx_\bool y$ iff $x=y$ and similarly for $\approx_\nat$, $(x,x')\approx_{\rho\times\tau} (y,y')$ iff $x\approx_\rho y$ and $y\approx_\tau y'$, $\seq{x_0,\ldots,x_{n-1}}\approx_{\sigma^\ast} \seq{y_0,\ldots,y_{m-1}}$ iff $n=m$ and $x_i\approx_{\sigma} y_i$ for all $i<n$ and finally $f\approx_{\rho\to\tau} g$ iff $fx\approx_\tau gx$ for all $x\in T_\rho$.

It turns out that all total objects are hereditarily extensional, in the sense that if $f\in T_{\rho\to\tau}$ and $x\approx_\rho y$ then $fx\approx_\tau fy$, and therefore the extensional collapse $C_\sigma:=T_\sigma/\approx_\sigma$ of the total objects constitutes a hierarchy $\modcont:=\{C_\sigma\}_{\sigma\in\types}$ of functionals in its own right. We call this hierarchy the \emph{total} continuous functionals, and as shown by Ershov, $\modcont$ is in fact isomorphic to the constructions of Kleene and Kreisel.

It is well known that $\modcont$ is a model of $\PAw$, and so in particular, any closed term $e:\sigma$ of System T has a canonical interpretation $e_C\in C_\sigma$, which can in turn be represented by some element $e_P\in T_\sigma$ of the corresponding equivalence class in $\modpar$. Suppose now that we extend System T with some new constant symbol $\Phi:\sigma$ which satisfies a recursive defining axiom
\begin{equation*}
(\ast) \ \ \ \Phi(x_1,\ldots,x_n)= r(\Phi,x_1,\ldots,x_n)
\end{equation*}
where $r$ is a closed term of System T. We can equivalently express $(\ast)$ as $\Phi=e(\Phi)$ for $e:\sigma\to\sigma$ defined by $$e(f):=\lambda x_1,\ldots ,x_n.r(f,x_1,\ldots,x_n).$$ Now, since $e$ is primitive recursive, it has a total representation $e_P\in T_{\sigma\to\sigma}\subset [P_\sigma\to P_\sigma]$, and it is a basic fact of domain theory that $\Phi$ can be given an interpretation $\Phi_P$ in $\modpar$ as a least fixed point of $e_P$ i.e.
\begin{equation*}
\Phi_P:=\bigsqcup_{n\in\NN} e_P^n(\bot_\sigma)
\end{equation*}
satisfies $\Phi_P=e_P(\Phi_P)$. If we can now show that $\Phi_P$ is in fact total, in other words that $\Phi_P(x_1,\ldots,x_n)$ is total for all total inputs $x_1,\ldots,x_n$, then defining $\Phi_C:=[\Phi_P]_{\approx_{\sigma}}\in C_\sigma$ we have
\begin{equation*}
\Phi_C=[\Phi_P]_{\approx_\sigma}=[e_P(\Phi_P)]_{\approx_\sigma}=[e_P]_{\approx_{\sigma\to\sigma}}[\Phi_P]_{\approx_\sigma}=e_C(\Phi_C)
\end{equation*}
and therefore the object $\Phi_C$ satisfies the defining axiom $(\ast)$ in $\modcont$. In other words, $\modcont$ is a model of the theory $\PAw+\Phi$, where by the latter we mean the extension of $\PAw$ with the new constant $\Phi$ and axiom $(\ast)$.

In short, in order to show that the extension of System T with some new form of recursion $\Phi$ is satisfied in $\modcont$, it suffices to show that the natural interpretation of $\Phi$ as a fixpoint in $\modpar$ is total. This approach has been widely used in the past to show that various forms of strong recursion arising from the axiom of choice have $\modcont$ as a model (see e.g. \cite[Proposition 5.1]{Berger(2004.0)} or \cite[Theorem 1]{BergOli(2005.0)}), and will be fundamental for us as well in Section \ref{sec-rec}.

In addition to showing that extensions of System T have a model, we must also confirm that they represent \emph{programs}, in the sense that any object of type $\nat$ can be effectively reduced to a numeral. This follows by appealing to Plotkin's adequacy theorem \cite{Plotkin(1977.0)}: We observe that terms of System T plus our new recursor $\Phi$ can be viewed as terms in PCF (recursion being dealt with by using the fixpoint combinator), which in addition inherit the usual call-by-value reduction semantics, with the defining axiom $(\ast)$ being interpreted as a rewrite rule. By showing that $\Phi$ represents a total object in the semantics of PCF within $\modpar$, it follows that any closed term $e:\nat$ in our extended calculus is denoted by some natural number i.e. $[e]\in \NN$, and by the adequacy theorem $e$ must then reduce to the numeral $\num{n}$. \smallskip

\noindent\textbf{Remark.} In order to avoid burdening ourselves with too many subscripts, in the remainder of this paper we use the same notation for $e:\sigma$ in $\PAw$, its canonical interpretation $e\in C_\sigma$ and some suitable representation $e\in P_\sigma$, rather than laboriously writing $e_C$ resp. $e_P$ whenever we are working in continuous models. Where there is any ambiguity, we make absolutely clear which system we are working in, and in the case of $e_P$ for primitive recursive $e$ we write explicitly how $e$ can be represented as a partial object unless this is obvious.
%
\section{A syntactic formulation of Zorn's lemma}
\label{sec-syntactic}

In this short section, we present a general axiomatic formulation of Zorn's lemma. This will be based on Theorem \ref{thm-open}, and is close in spirit to the axiom of \emph{open induction} as studied in \cite{Berger(2004.0)}. Like open induction, our axiom is of course weaker than the full statement of Zorn's lemma. Nevertheless, as we will see in Section \ref{sec-lex}, it in fact generalises open induction, and so in particular can be used to formalize highly non-trivial proofs such as Nash-Williams' minimal bad-sequence construction (cf. \cite{Berger(2004.0),Powell(2018.0)}). To be more specific, our axiom schema will take the shape of a maximum principle of the form
\begin{equation*}
\exists x P(x)\to \exists y (P(y)\wedge \forall z>y\neg P(z))
\end{equation*}
where $P(x)$ will range over formulas which are piecewise in the sense of Definition \ref{def-piece} and $<$ denotes some chain bounded partial order. However, our precise formulation of the axiom will be within the language of $\PAw$, and therefore both the notion of a piecewise formula and the relation $<$ need to be represented in a suitable way. \smallskip

\noindent\textbf{Remark.} From now all we annotate important definitions and results with the theory or model in which they take place, which will usually be some extension of $\PAw$ resp. $\HAw$ or one of $\modcont$ or $\modpar$.
\begin{definition}[$\PAw$/$\HAw$]
\label{def-papiece}
Suppose that $\comp{\cdot}{(\cdot)}:\sigma\times\delta\to\nu$ is a closed term of System T, and $Q(u^\nu)$ is a formula in the language of $\PAw$/$\HAw$. Then we say that the formula $P(x^\sigma):\equiv \forall d^\delta\; Q(\comp{x}{d})$ is piecewise w.r.t. $\comp{\cdot}{}$.
\end{definition}
Now, while it is too restrictive to demand that $<$ be represented by some primitive recursive functional $\sigma\times\sigma\to\bool$, for all applications we are interested in it suffices that $<$ can be expressed as a $\Sigma_1$ formula as follows:
\begin{equation*}
x<y:\equiv \exists a^\rho(y=_\sigma x\oplus a\wedge x\prec a)
\end{equation*}
where now $\oplus:\sigma\times\rho\to\sigma$ and $\prec:\sigma\times\rho\to\bool$ are closed terms of System T for some type $\rho$ (we use $x\oplus a$ to denote $\oplus(x,a)$ and $x\prec a$ to denote $\prec(x,a)=1$, and similarly $a\succ x$ to denote $x\prec a$). 
%
\begin{definition}[$\PAw$/$\HAw$]
\label{def-paopen}
Let $\comp{\cdot}{(\cdot)}:\sigma\times\delta\to\nu$, $\oplus:\sigma\times\rho\to\sigma$ and $\prec:\sigma\times\rho\to\bool$ be closed terms of System T. The axiom schema $\ZL_{\comp{}{},\oplus,\prec}$ is given by
\begin{equation*}
\begin{aligned}
\exists x^\sigma\forall d^\delta Q(\comp{x}{d})\to &\exists y^\sigma(\forall d\; Q(\comp{y}{d})\\
&\wedge \forall a\succ y\;\exists d\; \neg Q(\comp{y\oplus a}{d}))
\end{aligned}
\end{equation*}
where $Q(u^\nu)$ ranges over arbitrary formulas of $\PAw$ (and does not contain $x,y,a,d$ free).
\end{definition}

Note that our axiomatic formulation no longer mentions a main ordering $<$, which is instead induced by $\oplus$ and $\prec$. Note also that chain boundedness of $<$ is not formulated as a part of the axiom itself, and as such, validity of $\ZL_{\comp{}{},\oplus,\prec}$ in some given type structure will depend on the interpretation of $<$ being chain bounded in that model. We could of course seek to incorporate chain boundedness into the syntactic definition of Zorn's lemma and give a computational interpretation to the axiom as a whole. This would lead to a fascinating but extremely complex computational problem which would steer us in a quite different direction to the current article, and so we leave this to future work (cf. Section \ref{sec-conc}). We now illustrate our new principle by continuing our example from Section \ref{sec-prelim-zorn}, whose computational content has already been studied in in \cite{PowSchWie(2019.0)}.

\begin{example}
\label{ex-countableideal}
Let $\sigma:=\nat\to\bool$, $\delta:=\nat$, $\nu:=\bool^\ast$ and $\rho:=\nat\times (\nat\to\bool)$ and define
\begin{equation*}
\begin{aligned}
\comp{x}{d}&:=\seq{x(0),\ldots,x(d-1)}\\
x\oplus (n,y)&:=x\cup y\\
x\prec (n,y)&:= \bneg{x(n)}\cdot y(n)
\end{aligned}
\end{equation*}
where $\bneg{b}$ represents the negation of the boolean $b$ and 
\begin{equation*}
(x\cup y)(n):=\mbox{$1$ if $x(n)=1$ or $y(n)=1$ else $0$}.
\end{equation*}
These are all clearly definable as closed terms of System T, and in this case $\ZL_{\comp{}{},\oplus,\prec}$ is equivalent to
\begin{equation*}
\begin{aligned}
\exists x P(x)\to &\exists y(P(y)\\
&\wedge \forall (n,z)(\bneg{y(n)}=z(n)=1\to \neg P(y\cup z)))
\end{aligned}
\end{equation*}
for $P(x):\equiv \forall d\; Q(\seq{x(0),\ldots,x(d-1)})$. Here we can imagine objects $x:\nat\to\bool$ as characteristic functions for subsets of the natural numbers. Moreover, given some countable ring $R$ whose elements can be coded up as natural numbers and whose operations $+_R$ and $\cdot_R$ represented as primitive recursive functions $\nat\times\nat\to\nat$, the existence of a maximal ideal in $R$ would be provable in $\PAw+\ZL_{\comp{}{},\oplus,\prec}$. We do not give full details of this (an outline of the formalisation can be found in \cite{PowSchWie(2019.0)}). Instead we simply sketch why both $\modset$ and $\modcont$ satisfy $\ZL_{\comp{}{},\oplus,\prec}$ and are thus models of $\PAw+\ZL_{\comp{}{},\oplus,\prec}$.

Working in $\modcont$ (the same argument is also valid for $\modset$) we apply Theorem \ref{thm-open} for $X:=C_{\nat\to\bool}\cong \BB^\NN$ which via the identification of sets with their characteristic function is isomorphic to the powerset of $\NN$, together with the proper subset relation, observing that
\begin{equation*}
x\subset y\Leftrightarrow \exists (n,z)\in \NN\times\BB^\NN(y=x\cup z\wedge n\notin x\wedge n\in z)
\end{equation*}
where the right hand side is just the interpretation of the formula $\exists (n,z)(y=x\oplus (n,z)\wedge y\prec (n,z))$ in $\modcont$. Clearly $(X,\subset)$ is chain bounded w.r.t. $\comp{x}{d}:=\seq{x(0),\ldots,x(d-1)}$ using a simplified version of the argument in Example \ref{ex-ccapp}. Therefore for any formula $Q(\seq{u(0),\ldots,u(k-1)})$ in $\modcont$ on finite sequences of natural numbers the resulting formula $P(x):\equiv \forall d\; Q(\seq{x(0),\ldots,x(d-1)})$ on $X$ in $\modcont$ is piecewise w.r.t. $\comp{\cdot}{}$, and thus by Theorem \ref{thm-open} whenever $\exists x P(x)$ is satisfied there exists some $y\in X$ such that $P(y)$, and also $\neg P(z)$ whenever $y<z$ (or alternatively $\forall (n,z)(n\notin x\wedge n\in z\Rightarrow \neg P(y\cup z)$). Thus $\ZL_{\comp{}{},\oplus,\prec}$ is valid in $\modcont$.
\end{example}

\section{Recursion over chain bounded partial orders}
\label{sec-rec}

We now come to our first main contribution, in which we study modes of recursion over chain bounded partial orders that form an analogue to the axiom $\ZL_{\comp{}{},\oplus,\prec}$. A precise connection between a restricted form of $\ZL_{\comp{}{},\oplus,\prec}$ and our second mode of recursion will be presented in Section \ref{sec-dial}, but the results of this section are more general, and we consider them to be of interest in their own right. As such, this section could be read as a short, self-contained study in which we explore different recursion schemes over orderings induced by the parameters $(\oplus,\prec)$. Totality of our recursors will be justified using a variant of Theorem \ref{thm-open}, and the two main modes of recursion considered here will primarily differ in how we achieve `piecewise-ness' of the totality predicate. The first, which we characterise as `simple' recursion, uses a sequential continuity principle but is valid only for discrete output types, whereas the second, which we call `controlled' recursion, is total for arbitrary output type but uses an auxiliary parameter in the recursor itself to ensure wellfoundedness. 

For the remainder of this section, we fix types $\sigma,\rho,\delta$ and $\nu$, together with closed terms $\comp{\cdot}{}:\sigma\times\delta\to \nu$, $\oplus:\sigma\times\rho\to\sigma$ and $\prec:\sigma\times\rho\to\bool$ of System T, which are analogous to those in Section \ref{sec-syntactic}. For definitions and results below which take place in the model $\modpar$, note that $\comp{\cdot}{}\in T_{\sigma\times\delta\to\nu}$ denotes some canonical representation of the corresponding term of System T as a total continuous functional, and similarly for $\oplus\in T_{\sigma\times\rho\to\sigma}$ and $\prec\in T_{\sigma\times\rho\to\bool}$ (cf. Section \ref{sec-prelim-cont}).

\subsection{Simple recursion over $(\oplus,\prec)$}
\label{sec-rec-simple}

The first recursion scheme we consider is represented by the constant $\Phi^\theta_{\oplus,\prec}$ equipped with defining equation 
\begin{equation}
\label{eqn-simple}
\Phi fx=_\theta fx(\lambda a \; . \; \Phi f(x\oplus a)\mbox{ if $a\succ x$ else $0_\theta$})
\end{equation}
where $f:\sigma\to (\rho\to\theta)\to\theta$ and $x:\sigma$, and we recall that $0_\theta$ is a canonical zero term of type $\theta$. Note that in the defining equation we suppressed the parameters on $\Phi$ - and we will continue to do this whenever there is no risk of ambiguity. In what follows, it will be helpful to use the abbreviation 
\begin{equation*}
\Phi_{f,x}:=\lambda a \; . \; \Phi f(x\oplus a)\mbox{ if $a\succ x$ else $0_\theta$}
\end{equation*}
so that the defining equation can then be expressed as 
\begin{equation*}
\Phi fx=fx\Phi_{f,x}.
\end{equation*}

%
\begin{definition}[$\modpar$]
\label{def-cont}
Let $L\subseteq T_\sigma$. We say that a functional $\psi\in P_{\sigma\to\theta}$ is \emph{piecewise continuous} with respect to $\comp{\cdot}{}$ and $L$, if for any $x\in L$ such that $\psi x\in T_\theta$ there exists some $d\in T_\delta$ such that 
\begin{equation*}
\forall y\in T_\sigma(\comp{x}{d}=\comp{y}{d}\Rightarrow \psi y\in T_\theta).
\end{equation*}
\end{definition}
\begin{definition}[$\modpar$]
\label{def-compatible}
A partial order $<$ on $T_\sigma$ is compatible with $(\oplus,\prec)$ if $x<x\oplus a$ for any $(x,a)\in T_{\sigma\times\rho}$ with $x\prec a$ (i.e. $\prec(x,a)=_\BB 1$).
\end{definition}
The next definition is a slight adaptation of Definition \ref{def-ccapp}, where now we require $\ub{\gamma}$ to be an element of some subset $L$ of the main partial order.
\begin{definition}[$\modpar$]
\label{def-ccpartial}
A partial order $<$ on $T_\sigma$ is chain bounded with respect to $\comp{\cdot}{}$ and $L\subseteq T_\sigma$ if every nonempty chain $\gamma\subseteq T_\sigma$ has an upper bound $\ub{\gamma}\in L$ satisfying the property that for any $d\in T_\delta$ there exists some $x \in\gamma$ such that $\comp{\ub{\gamma}}{d}=\comp{x}{d}$.
\end{definition}
We now come to our first totality result. This establishes a condition on inputs $f$ which ensures totality of $\Phi f$. As we will see, in certain natural situations we can use this to show that $\Phi f$ is total for \emph{any} total $f$, and thus $\Phi$ itself is total. However, our result is more general as it also allows us to establish totality of $\Phi f$ in cases where $\Phi$ may not be.
\begin{theorem}[$\modpar$]
\label{thm-simpletot}
Let $\Phi$ denote the least fixed point of the primitive recursive defining equation (\ref{eqn-simple}), and suppose that there exist $<$ on $T_\sigma$ and $L\subseteq T_\sigma$ such that $<$ is compatible with $(\oplus,\prec)$ and chain bounded w.r.t. $\comp{\cdot}{}$ and $L$. Let $f\in T_{\sigma\to (\rho\to\theta)\to\theta}$. Then whenever $\Phi f\in P_{\sigma\to\theta}$ is piecewise continuous w.r.t. $\comp{\cdot}{}$ and $L$, it follows that $\Phi f\in T_{\sigma\to\theta}$.
\end{theorem}
\begin{proof}
By Zorn's lemma. Suppose that $\Phi f$ is piecewise continuous, and consider the set $S\subseteq T_\sigma$ given by
\begin{equation*}
S:=\{x\in T_\sigma \; | \; \Phi fx\notin T_\theta\}.
\end{equation*}
We first show that $(S,<)$ is chain bounded in the usual sense. Taking some nonempty chain $\gamma\subseteq S$, by chain boundedness in the sense of Definition \ref{def-ccpartial} this has some upper bound $\ub{\gamma}\in L$. Suppose for contradiction that $\Phi f\ub{\gamma}\in T_\theta$. By piecewise continuity of $\Phi f$ there exists some $d\in T_\delta$ such that $\comp{\ub{\gamma}}{d}=\comp{y}{d}$ implies $\Phi fy\in T_\theta$ for any $y\in T_\sigma$. But then there exists some $x\in \gamma$ with $\comp{\ub{\gamma}}{d}= \comp{x}{d}$ and thus $\Phi fx\in T_\theta$, contradicting $x\in S$. Therefore $\Phi f\ub{\gamma}\notin T_\theta$ and thus $\ub{\gamma}\in S$.

To prove the main result, suppose for contradiction that $\Phi f\notin T_{\sigma\to\theta}$, which implies that $S\neq\emptyset$. Then by Zorn's lemma, $S$ has some maximal element $x$. But for any $a\in T_\rho$ with $x\prec a$ we have $x<x\oplus a$ by compatibility, and thus $\Phi_{f,x}(a)=\Phi f(x\oplus a)\in T_\theta$. It follows that $\Phi_{f,x}\in T_{\rho\to\theta}$, since in the other case $\neg (x\prec a)$ we have $\Phi_{f,x}(a)=0_\theta\in T_\theta$. But then by totality of $f$ we have $\Phi fx=fx\Phi_{f,x}\in T_\theta$, contradicting $x\in S$. Therefore $S=\emptyset$ and so $\Phi f\in T_{\sigma\to\theta}$.
\end{proof}
The technique we have used in this proof is a generalisation of the proof of Theorem 0.3 from \cite{Berger(2002.0)}, which uses Zorn's lemma to show that the so-called Berardi-Bezem-Coquand functional defined in \cite{BBC(1998.0)} is total. We now give a concrete example of how the result can be applied, but first we state and prove a sequential continuity lemma (cf. also \cite[Lemma 0.1]{Berger(2002.0)}), which will also be useful in later sections.
\begin{lemma}[$\modpar$]
\label{lem-seqcont}
Let $\theta$ be a \emph{discrete} type i.e. one which does not contain function types. Suppose that $\psi\in P_{(\nat\to\sigma)\to\theta}$ where $\sigma$ is some arbitrary type, that $x\in T_{\nat\to\sigma}$ satisfies $x(\bot)=\bot_\sigma$ and that $\psi x\in T_\theta$. Then there is some $d\in\NN$ such that for any $y\in P_{\nat\to\sigma}$, whenever $x(i)=y(i)$ for all $i<d$ then $\psi x=\psi y$. 
\end{lemma}
\begin{proof}
We use a simple adaptation of the proof of Lemma 0.1 of \cite{Berger(2002.0)}.
Since $T_\theta$ is open in the Scott topology whenever $\theta$ is discrete, there is some compact $x_0\sqsubseteq x$ such that $\psi x_0\in T_\theta$, and since $\psi x_0\sqsubseteq \psi x$ we must in fact have $\psi x_0=\psi x$ (that $y\sqsubseteq z$ implies $y=z$ for $y\in T_\theta$ is evidently true for $\theta=\NN_\bot$ or $\theta=\BB_\bot$, and holds for arbitrary discrete $\theta$ by induction over its structure). Now, since $x_0$ is compact (i.e. contains only a finite amount of information) there is some $d\in\NN$ such that $x_0(i)=\bot_\sigma$ for all $i\geq d$. Suppose now that $y\in P_{\nat\to\sigma}$ satisfies $x(i)=y(i)$ for all $i<d$. We claim that $x_0\sqsubseteq y$. To see this, note that for $i<d$ we have $x_0(i)\sqsubseteq x(i)=y(i)$, for $i\geq d$ we have $x_0(i)=\bot_\sigma\sqsubseteq y(i)$, and for $i=\bot$ since $x(\bot)=\bot_\sigma$ we must also have $x_0(\bot)=\bot_\sigma\sqsubseteq y(\bot)$. Therefore $\psi x_0\sqsubseteq \psi y$ and since $\psi x_0\in T_\theta$ we must have $\psi y=\psi x_0=\psi x$.
\end{proof}
\begin{example}
\label{ex-simpleset}
Let $\comp{\cdot}{}$, $\oplus$ and $\prec$ be the obvious total representatives of the primitive recursive functions defined in Example \ref{ex-countableideal} i.e. extensions that are defined also on non-total input, for example
\begin{equation*}
\begin{aligned}
\comp{x}{d}&:=\seq{x(0),\ldots,x(d-1)}\mbox{ for $d\in\NN$ else $\bot$}
\end{aligned}
\end{equation*}
and so on. We observe that $T_\sigma=T_{\nat\to\bool}$ is the set of all functions $x:\NN_\bot\to\BB_\bot$ which are monotone (in the domain theoretic sense) and satisfy $x(n)\in\BB$ whenever $n\in\NN$. We define $L\subset T_{\nat\to\bool}$ to consist of those functions which are \emph{strict}, in that they satisfy in addition $x(\bot)=\bot$.

Now suppose that $\theta$ is discrete. Then \emph{any} function $\psi\in P_{(\nat\to\bool)\to \theta}$ is piecewise continuous w.r.t. $\comp{\cdot}{}$ and $L$. To see this, take any strict $x$ such that $\psi x\in T_\theta$. Then by Lemma \ref{lem-seqcont} there exists some $d\in\NN$ such that for any $y\in P_{\nat\to\bool}$ (and so in particular $y\in T_{\nat\to\bool}$) we have $\psi y=\psi x\in T_\theta$ whenever $\comp{x}{d}=\seq{x(0),\ldots,x(d-1)}=\seq{y(0),\ldots,y(d-1)}=\comp{y}{d}$.

Next define $<$ on $T_{\nat\to\bool}$ by $x<y$ iff $x(i)=1\Rightarrow y(i)=1$ for all $i\in\NN$ and there exists at least one $j\in\NN$ with $x(j)=0$ and $y(j)=1$. Then $<$ is compatible with $(\oplus,\prec)$, and moreover, for any nonempty chain $\gamma\subseteq T_{\nat\to\bool}$ define $\ub{\gamma}\in L$ by
\begin{equation*}
\ub{\gamma}(n):=\begin{cases}1 & \mbox{if $x(n)=1$ for some $x\in\gamma$}\\ 0 & \mbox{otherwise}\end{cases}
\end{equation*}
and $\ub{\gamma}(\bot)=\bot$. Then clearly $x\leq \ub{\gamma}$ for all $x\in\gamma$, and moreover for any $d\in\NN=T_{\nat}$, by a variant of the argument in Example \ref{ex-ccapp} we have $\comp{\ub{\gamma}}{d}=\comp{x}{d}$ for some $x\in \gamma$. Thus $<$ is chain bounded w.r.t. $\comp{\cdot}{}$ and $L$.

Now, let $\Phi$ denote the least fixed point in $\modpar$ of
\begin{equation}
\label{eqn-update}
\Phi fx=_\theta fx(\lambda (n,y)\; . \; \Phi f(x\cup y)\mbox{ if $\bneg{x(n)}\cdot y(n)=1$ else $0_\theta$}).
\end{equation}
By Theorem \ref{thm-simpletot}, taking any total $f$, since $\Phi f\in P_{(\nat\to\bool)\to\theta}$ is automatically piecewise continuous, we have that $\Phi f$ is total, and therefore $\Phi$ is a total object in $\modpar$. This implies that $\modcont\models \PAw+\Phi$ for the extension of $\PAw$ with some new constant satisfying the defining axiom (\ref{eqn-update}).
\end{example}

\subsection{On non-discrete output types}
\label{sec-rec-problem}

In Example \ref{ex-simpleset} we have essentially shown that a simple variant of `update induction' in the sense of \cite{Berger(2004.0)} is total. In fact, with a slight modification of the above proof we would be able to reprove totality of update induction in its general form. However, in this paper we are primarily interested in forms of recursion on chain bounded partial orders which \emph{do not} correspond to the simple recursive scheme (\ref{eqn-simple}). The reason for this is that in order to establish totality of $\Phi f$ for any total $f$, we are typically required to restrict the complexity of the output type $\theta$ to being discrete, so that something along the lines of Lemma \ref{lem-seqcont} applies. As we will see, this is a problem for the functional interpretation.

Before we go on, we illustrate why extending $\PAw/\HAw$ with (\ref{eqn-update}) for non-discrete types does not even result in a consistent theory! Let us set $\theta:=\nat\to\nat$ and define $f$ by
\begin{equation*}
fxp:=\lambda n\; . \; 1+p(n,\{n\})(n+1)
\end{equation*}
where we identify the set $\{n\}$ with its characteristic function of type $\nat\to\bool$. Then defining $k:=\Phi f(\emptyset)0:\nat$ we have
\begin{equation*}
\begin{aligned}
k&=1+\Phi f({\{0\}})(1)=2+\Phi f({\{0,1\}})(2)\\
&=\ldots =k+1+\Phi f({\{0,\ldots,k\}})(k+1)>k
\end{aligned}
\end{equation*}
which is inconsistent with $\PAw/\HAw$. The key point at which the argument from Example \ref{ex-simpleset} fails is that Lemma \ref{lem-seqcont} is no longer valid for $\theta:=\nat\to\nat$: if $\psi x$ is a function then it can in general query an infinite part of $x$. To overcome this we could restrict our attention to those $f$ such that $\Phi f$ is piecewise continuous and thus total for non-discrete output type: For example, let
\begin{equation*}
fxp:=\lambda n<\num{N}\; . \; 1+p(n,\{n\})(n+1)
\end{equation*}
for some numeral $\num{N}:\nat$, so that $p$ is only queried finitely many times. Then working in $\modpar$, for total $x$ we would have
\begin{equation*}
\Phi fx=\lambda n\; . \; 1+\begin{cases}\Phi f(x\cup\{n\})(n+1) & \mbox{if $x(n)=1\wedge n<N$}\\ 0 & \mbox{otherwise} \end{cases}
\end{equation*}
and so a point of continuity for $\Phi fx$ could be taken to be the maximum of all points of continuity of the functions $\lambda y.\Phi f(y\cup\{n\})(n+1)$ for $n<N$ and at point $y:=x$. 

We now propose cleaner way of extending (\ref{eqn-update}) to non discrete output types. Instead of restricting $f$, we add a new parameter $\omega$ which controls the recursion directly.
%
\subsection{Controlled recursion over $(\oplus,\prec)$}
\label{sec-rec-controlled}

We modify the scheme (\ref{eqn-simple}), resulting in a slightly more elaborate mode of recursion in which the continuity behaviour is controlled by some auxiliary functional $\omega$. Define the constant $\Psi^\theta_{\oplus,\prec}$ (from now on omitting the parameters) by
\begin{equation}
\label{eqn-controlled}
\begin{aligned}
&\Psi \omega fx=_\theta \\
&f\cont{x}{\Psi}{\omega,f}(\lambda a\; . \; \Psi \omega f(\cont{x}{\Psi}{\omega,f}\oplus a)\mbox{ if $a\succ \cont{x}{\Psi}{\omega,f}$ else $0_\theta$})
\end{aligned}
\end{equation}
where $f:\sigma\to (\rho\to\theta)\to\theta$ and $\omega:\sigma\to (\rho\to\theta)\to\sigma$ and $\cont{x}{\Psi}{\omega,f}$ is defined by
\begin{equation*}
\label{eqn-omega}
\cont{x}{\Psi}{\omega,f}:=_\sigma\omega x(\lambda a\; . \; \Psi\omega f(x\oplus a)\mbox{ if $a\succ x$ else $0_\theta$}).
\end{equation*}
Observe that $\Psi$ is still defined as the fixed point of a simple closed term of $\PAw$, and as it will turn out, this modified scheme will allow us to admit output of arbitrary type level. Moreover, we will show later that by instantiating $\omega$ by a suitable closed term of $\PAw$, we can use this recursive scheme to define a realizer for the functional interpretation of our axiomatic form of Zorn's lemma. As before, we use the abbreviation
\begin{equation*}
\Psi_{\omega,f,x}:=\lambda a\; . \; \Psi\omega f(x\oplus a)\mbox{ if $a\succ x$ else $0_\theta$}
\end{equation*}
so that the defining equation (\ref{eqn-controlled}) now becomes
\begin{equation*}
\Psi\omega fx=f\cont{x}{\Psi}{\omega,f}\Psi_{\omega,f,\cont{x}{\Psi}{\omega,f}}
\end{equation*}
for $\cont{x}{\Psi}{\omega,f}:=\omega x\Psi_{\omega,f,x}$. We now give a totality theorem analogous to Theorem \ref{thm-simpletot}, but with the notion of piecewise continuity replaced by a slightly more subtle property.
\begin{definition}[$\modpar$]
\label{def-trunc}
We say that a pair of functionals $\psi\in P_{\sigma\to\sigma}$ and $\phi\in P_{\sigma\to\theta}$ form a truncation with respect to $\comp{\cdot}{}$, $L\subseteq T_\sigma$ and some partial order $<$ on $T_\sigma$ if the following two conditions are satisfied:
\begin{enumerate}[(a)]

\item\label{truncb} For any $x,y\in T_\sigma$, if $\psi x\in T_\sigma$ and $\psi x<y$ then $x<y$.

\item\label{trunca} For any $x\in L$ such that $\psi x\in T_\sigma$ and $\phi(\psi x)\in T_\theta$ there exists some $d\in T_\delta$ such that
\begin{equation*}
\forall y\in T_\sigma(\comp{x}{d}=\comp{y}{d}\Rightarrow \psi x=\psi y).
\end{equation*}

\end{enumerate} 
\end{definition}
\begin{example}
\label{ex-trunc}
Continuing from Example \ref{ex-simpleset} with $\sigma:=\nat\to\bool$ but $\theta$ now arbitrary, for any $N\in\NN$ the continuous functional $\psi_N\in T_{\sigma\to\sigma}$ defined by
\begin{equation*}
\psi_Nx(n):=\begin{cases}x(n) & \mbox{if $n<N$}\\ 1 & \mbox{otherwise}\end{cases}
\end{equation*}
and $\psi_Nx(\bot)=\bot$ forms a truncation with any other functional $\phi\in P_{\sigma\to\theta}$ w.r.t $\comp{\cdot}{}$, $L$ and $<$. To see this, observe that for any strict $x$, if $y\in T_\sigma$ satisfies $\seq{x(0),\ldots,x(N-1)}=\seq{y(0),\ldots,y(N-1)}$ then $\psi_Nx=\psi_Ny$ and so $\psi_N$ satisfies condition (\ref{trunca}) of being a truncation for $d:=N$. For condition (\ref{truncb}), if $\psi_Nx<y$ this means that $x(i)=1\Rightarrow y(i)=1$ for all $i<N$, $y(i)=1$ for all $i\geq N$, and $x(j)=0$ and $y(j)=1$ for some $j<N$, from which it follows easily that $x<y$. 

On the other hand, the continuous functional $\phi_N\in T_{\sigma\to\sigma}$ defined by
\begin{equation*}
\phi_Nx(n):=\begin{cases}x(n) & \mbox{if $n<N$}\\ 0 & \mbox{otherwise}\end{cases}
\end{equation*}
and $\phi_Nx(\bot)=\bot$ satisfies condition (\ref{trunca}) of being a truncation, but not condition (\ref{truncb}), since for $x$ representing the characteristic function of the singleton set $\{N\}$ we have $\phi_Nx<x$ but not $x<x$. Finally, the identity functional $\iota\in T_{\sigma\to\sigma}$ clearly satisfies condition (\ref{truncb}) of being a truncation, but condition (\ref{trunca}) fails for e.g. $\phi$ also the identity function, since for some arbitrary $x\in L$ there is no $d\in\NN$ such that for any total $y$, $\seq{x(0),\ldots,x(d-1)}=\seq{y(0),\ldots,y(d-1)}$ implies $x(n)=y(n)$ for all $n\in\NN$.

\end{example}
\begin{theorem}[$\modpar$]
\label{thm-totcont}
Let $\Psi$ denote the least fixed point of the primitive recursive defining equation (\ref{eqn-controlled}), and suppose that there exist $<$ on $T_\sigma$ and $L\subseteq T_\sigma$ such that $<$ is compatible with $(\oplus,\prec)$ and chain bounded w.r.t. $\comp{\cdot}{}$ and $L$. Let $f\in T_{\sigma\to (\rho\to\theta)\to\theta}$ and $\omega\in T_{\sigma\to (\rho\to\theta)\to\sigma}$. Then whenever $\cont{\cdot}{\Psi}{\omega,f}\in P_{\sigma\to\sigma}$ and $\lambda x\; . \; fx\Psi_{\omega,f,x}\in P_{\sigma\to\theta}$ form a truncation w.r.t. $\comp{\cdot}{}$, $L$ and $<$, it follows that $\Psi\omega f\in T_{\sigma\to\theta}$.
\end{theorem}
\begin{proof}
We again appeal to Zorn's lemma, but this time on the set $S\subseteq T_\sigma$ given by
\begin{equation*}
S:=\{x\in T_\sigma \; | \; \mbox{either $\cont{x}{\Psi}{\omega,f}\notin T_\sigma$ or $\Psi\omega fx\notin T_\theta$}\}.
\end{equation*}
To show that $(S,<)$ is chain bounded in the usual sense, take some nonempty chain $\gamma\subseteq S$ and consider its upper bound $\ub{\gamma}\in L$ in the sense of Definition \ref{def-ccpartial}. As before, we want to show that $\ub{\gamma}\in S$, so we assume for contradiction that this is not the case, which means that both $\cont{\ub{\gamma}}{\Psi}{\omega,g}\in T_\sigma$ and $\Psi\omega f\ub{\gamma}\in T_\theta$. But then by Definition \ref{def-trunc} (\ref{trunca}) - observing that $\cont{\ub{\gamma}}{\Psi}{\omega,g}$ and $f\cont{\ub{\gamma}}{\Psi}{\omega,g}\Psi_{\omega,f,\cont{\ub{\gamma}}{\Psi}{\omega,g}}=\Psi\omega f \ub{\gamma}$ are both total - there exists some $d\in T_\delta$ such that $\cont{\ub{\gamma}}{\Psi}{\omega,f}=\cont{y}{\Psi}{\omega,f}$ for any $y\in T_\sigma$ satisfying $\comp{\ub{\gamma}}{d}=\comp{y}{d}$. But since by Definition \ref{def-ccpartial} there exists some $x\in\gamma$ such that $\comp{\ub{\gamma}}{d}=\comp{x}{d}$ it therefore follows that $\cont{\ub{\gamma}}{\Psi}{\omega,f}=\cont{x}{\Psi}{\omega,f}$ and thus
\begin{equation*}
\Psi\omega f\ub{\gamma}=f\cont{\ub{\gamma}}{\Psi}{\omega,f}\Psi_{\omega,f,\cont{\ub{\gamma}}{\Psi}{\omega,f}}=f\cont{x}{\Psi}{\omega,f}\Psi_{\omega,f,\cont{x}{\Psi}{\omega,f}}=\Psi\omega fx
\end{equation*} 
which imply that $\cont{x}{\Psi}{\omega,f}\in T_\sigma$ and $\Psi\omega fx\in T_\theta$ and thus $x\notin S$, a contradiction. Thus $\ub{\gamma}\in S$ and $S$ is chain bounded.

We now suppose that the conclusion of the main result is false, which means that there exists some $x\in T_\sigma$ such that $\Psi\omega fx\notin T_\theta$, and so in particular $x\in S$ and thus $S$ is nonempty. By Zorn's lemma, $S$ contains a maximal element $x$. We now show that $x\notin S$, a contradiction. Since $x$ is maximal, for any $x<y$ we must have $\cont{y}{\Psi}{\omega,f}\in T_\sigma$ and $\Psi \omega fy\in T_\theta$. 

We first show that $\Psi_{\omega,f,x}$ is total: For any $a\in T_\rho$, either $x\succ a$ and so $\Psi_{\omega,f,x}a=0_\theta\in T_\theta$, or $x\succ a$ and thus by compatibility we have $x<x\oplus a$ and therefore $\Psi_{\omega,f,x}a=\Psi\omega f(x\oplus a)\in T_\theta$. But then since $\omega$, $x$ and $\Psi_{\omega,f,x}$ are all total, it follows that $\cont{x}{\Psi}{\omega,f}=\omega x\Psi_{\omega,f,x}\in T_\sigma$. 

We now show that $\Psi_{\omega,f,\cont{x}{\Psi}{\omega,f}}$ is total: For $a\in T_\rho$, either $\cont{x}{\Psi}{\omega,f}\succ a$ and so $\Psi_{\omega,g,\cont{x}{\Psi}{\omega,f}}a=0_\theta\in T_\theta$, or $\cont{x}{\Psi}{\omega,f}\succ a$ and thus by compatibility we have $\cont{x}{\Psi}{\omega,f}<\cont{x}{\Psi}{\omega,f}\oplus a$. But now using condition (\ref{truncb}) of $\cont{\cdot}{\Psi}{\omega,f}$ forming a truncation, we have $x<\cont{x}{\Psi}{\omega,f}\oplus a$ and thus $\Psi \omega f(\cont{x}{\Psi}{\omega,f}\oplus a)\in T_\theta$. Now, since $f$, $\cont{x}{\Psi}{\omega,f}$ and $\Psi_{\omega,f,\cont{x}{\Psi}{\omega,f}}$ are all total, it follows that $\Psi\omega fx=f\cont{x}{\Psi}{\omega,f}\Psi_{\omega,f,\cont{x}{\Psi}{\omega,f}}\in T_{\theta}$.

We have therefore proven that if $x$ is maximal, then both $\cont{x}{\Psi}{\omega,f}$ and $\Psi\omega fx$ are total and so $x\notin S$, contradicting that $S$ has a maximal element. Therefore $S=\emptyset$ and so $\Psi\omega fx\in T_\theta$ for any $x\in T_\sigma$ and we have shown totality of $\Psi \omega f$.
\end{proof}
\begin{example}
\label{ex-controlledset}
We now consider Example \ref{ex-simpleset} from the perspective of controlled recursion, using a truncation similar to that given in Example \ref{ex-trunc} above. Let us extend the language of $\PAw$ with a new constant $\Omega$ with defining equation
\begin{equation}
\label{eqn-controlledex}
\Omega nfx=f\conte{x}{n}(\lambda a\; . \; \Omega nf(\conte{x}{n}\oplus a)\mbox{ if $a\succ \conte{x}{n}$ else $0$})
\end{equation}
where $n:\nat$, $f:\sigma\to (\rho\to\theta)\to\theta$ and $\conte{x}{n}$ is defined by
\begin{equation*}
\conte{x}{n}:=_{\nat\to\bool}\lambda i\; . \; x(i)\mbox{ if $i<n$ else 1}.
\end{equation*}
Then $\Omega$ is definable as $\Omega nfx:=\Psi(c n)fx$ where $\Psi$ satisfies (\ref{eqn-controlled}) and $c:\nat\to\sigma\to (\rho\to\theta)\to\sigma$ is the primitive recursive functional defined by
\begin{equation*}
cnxp:=\lambda i\; . \; x(i)\mbox{ if $i<n$ else 1}.
\end{equation*}
Working from now on in $\modpar$, for each $n\in\NN$ we can interpret $\Omega n$ as being a least fixed point of the equation (\ref{eqn-controlled}) for $\omega$ instantiated as the total representation in $\modpar$ of $cn$ as above. We apply Theorem \ref{thm-totcont} to show that $\Omega n$ is total. The compatibility and chain boundedness requirements are the same as in Theorem \ref{thm-simpletot}, and so in our setting have already been dealt with in Example \ref{ex-simpleset}. Take $f\in T_{\sigma\to (\rho\to\theta)\to \theta}$ with $\omega:=cn\in T_{\sigma\to (\rho\to\theta)\to\sigma}$. To see that $\cont{\cdot}{\Psi}{\omega,f}$ and $\lambda x\; . \; fx\Psi_{\omega,f,x}$ form a truncation, we use a similar argument to Example \ref{ex-trunc}. We can assume that $\omega=cn$ is interpreted in $\modpar$ as
\begin{equation*}
\omega xp(i)=\begin{cases}x(i) & \mbox{if $i<n$}\\ 1 & \mbox{otherwise}\end{cases}
\end{equation*}
and $\omega xp(\bot)=\bot$, and so for any strict $x$, if $y\in T_\sigma$ satisfies $\seq{x(0),\ldots,x(n-1)}=\seq{y(0),\ldots,y(n-1)}$ then $\omega xp=\omega yp=\omega yq$ for any functionals $p,q\in P_{\rho\to\theta}$, and so in particular $\cont{x}{\Psi}{\omega,f}=\omega x\Psi_{\omega,f,x}=\omega y\Psi_{\omega,f,y}=\cont{y}{\Psi}{\omega,f}$. This establishes property (\ref{trunca}), and property (\ref{truncb}) follows analogously to Example \ref{ex-trunc}. Therefore $\Psi \omega f=\Psi(cn)f\in T_{\sigma\to\theta}$ for arbitrary $n\in\NN$ and $f\in T_{\sigma\to (\rho\to\theta)\to\theta}$, which implies that the object $\Omega$ defined by $\Omega nfx:=\Psi(cn)fx$ is total and satisfies the equation (\ref{eqn-controlledex}) in $\modpar$. Therefore $\Omega$ also has an interpretation in $\modcont$, i.e. $\modcont\models \PAw+\Omega$. Note that no conditions were imposed on $\theta$, and so totality of $\Omega$ also holds for non-discrete $\theta$.
\end{example}
The functional defined in (\ref{eqn-controlledex}) is rather strongly controlled by $cn$ (we claim in fact that $\Omega$ is definable as a term of System T). This deliberately simplistic example was chosen simply to illustrate Theorem \ref{thm-totcont}. Nevertheless, later we will require a much more subtle truncation for realizing the functional interpretation of lexicographic induction, which certainly does lead us beyond the realm of primitive recursion.

\section{The functional interpretation of Zorn's lemma}
\label{sec-dial}

In the last section we introduced two general variants of recursion over chain bounded partial orders. We will now show that our controlled variant is well suited for solving the functional interpretation of our axiomatic formulation of Zorn's lemma from Section \ref{sec-syntactic}. We begin by recalling some essential facts about the functional interpretation. Full details can be found in \cite{AvFef(1998.0)} or \cite[Chapters 8 \& 10]{Kohlenbach(2008.0)}. For those readers not familiar with the functional interpretation, we aim to at least present, in a self contained manner, the concrete \emph{computational problem} we need to solve. This alone should suffice in order to understand later sections. Such a reader is advised to skip directly ahead to Section \ref{sec-dial-solve} (perhaps skimming through Section \ref{sec-dial-zorn} on the way).

\subsection{An overview of the functional interpretation}
\label{sec-dial-overview}

So that we can formally state a higher-type variant of G\"odel's soundness theorem for the functional interpretation, we need to recall the so-called \emph{weakly extensional} variant $\WPAw$ of $\PAw$, which is obtained from the latter by simply replacing the axiom of extensionality with a quantifier-free rule form (see \cite[Definition 3.12]{Kohlenbach(2008.0)} for details, though this is not necessary to understand what follows). This is because the interpretation is unable to deal with the axiom of extensionality and thus cannot be applied directly to $\PAw$ (see \cite[pp. 15]{AvFef(1998.0)} or \cite[pp. 126--127]{Kohlenbach(2008.0)}).  

The functional interpretation assigns to each formula $A$ of $\WPAw$ a new formula $\dt{A}{x}{y}$ where now $x$ and $y$ are (possibly empty) tuples of variables of some finite type. The precise definition is by induction over the logical structure of $A$, and is given in Figure 1, where in the interpretation of disjunction, $i$ is an object of natural number type and $P\vee_i Q$ denotes $(i= 0\to P)\wedge (i\neq 0\to Q)$. The basic functional interpretation applies only to intuitionistic theories. In order to deal with classical logic, we need to combine the interpretation with some variant of the negative translation $A\mapsto A^N$ as an initial step. We do not give any further details, but simply state the main soundness theorem for classical arithmetic. In the following, $\QFAC$ denotes the axiom of quantifier-free choice i.e. the schema
\begin{equation*}
\forall x^\rho\exists y^\sigma\; A_0(x,y)\to \exists f^{\rho\to\sigma}\forall x\; A_0(x,fx)
\end{equation*}
where $\rho$ and $\sigma$ are arbitrary types and $A_0(x,y)$ ranges over quantifier-free formulas.

\begin{figure}[t]
\label{fig-dt}
\begin{gather*}\dt{A}{}{}:\equiv A\mbox{ for $A$ atomic} \ \ \ \dt{A\wedge B}{x,u}{y,v}:\equiv \dt{A}{x}{y}\wedge\dt{B}{u}{v}\\
\dt{A\vee B}{i,x,u}{y,v}:\equiv \dt{A}{x}{y}\vee_i\dt{B}{u}{v}\\ \dt{A\to B}{U,Y}{x,v}:\equiv \dt{A}{x}{Yxv}\to \dt{B}{Ux}{v}\\
\dt{\exists z A(z)}{x,u}{v}:\equiv \dt{A(x)}{u}{v} \ \ \ \dt{\forall z A(z)}{U}{x,v}:\equiv \dt{A(x)}{Ux}{v}
\end{gather*}\caption{The functional interpretation}\end{figure}
\begin{theorem}[cf. Theorem 10.7 of \cite{Kohlenbach(2008.0)}, but essentially due to G\"odel \cite{Goedel(1958.0)}]
\label{thm-dialectica}
Let $A$ be a formula in the language of $\WPAw$. Then whenever
\begin{equation*}
\WPAw+\QFAC\vdash A
\end{equation*}
we can extract a term $t$ of $\WHAw$ whose free variables are the same as those of $A$, and such that
\begin{equation*}
\HAw\vdash \forall y\dt{A^N}{t}{y}.
\end{equation*}
\end{theorem}

\noindent\textbf{Remark.} Note that $\forall y\dt{A^N}{t}{y}$ is provable even in a quantifier-free fragment of $\WHAw$, the intuitionistic variant of $\WPAw$.\smallskip

The main result in the remainder of this section is to extend Theorem \ref{thm-dialectica} above to include our formulation of Zorn's lemma. Generally speaking, in order to expand the soundness theorem to incorporate extensions of $\WPAw+\QFAC$ with new axioms $X$, it suffices to provide a new recursive scheme $\Omega$ such that the functional interpretation of $X^N$ has a solution in $\HAw+\Omega$. A classical example of this is with $X$ as the axiom of countable choice, and $\Omega$ the scheme of bar recursion in all finite types (cf. \cite[Chapter 6]{AvFef(1998.0)} or \cite[Chapter 11]{Kohlenbach(2008.0)}). Here on the other hand, we set $X$ to be our syntactic formulation of Zorn's lemma, and $\Omega$ a functional definable from our scheme of controlled recursion from Section \ref{sec-rec-controlled}.
%
\subsection{The functional interpretation of $\ZL_{\comp{}{},\oplus,\prec}$}
\label{sec-dial-zorn}

We now outline how the combination of the functional interpretation with the negative interpretation acts on the axiom $\ZL_{\comp{}{},\oplus,\prec}$ as given in Definition \ref{def-paopen}, subject to the additional restriction that $Q(u)$ ranges over \emph{quantifier-free} formulas of $\WPAw$ (similar restrictions can be found in \cite{Berger(2004.0)} and \cite{Coquand(1991.0)} in the context of open induction). This restriction still allows us to deal with most concrete examples we are interested in (including the existence of maximal ideals in countable commutative rings in Example \ref{ex-countableideal} and also Higman's lemma, which we will discuss later), but simplifies the interpretation considerably (though we conjecture that in many cases, and in particular for concrete example considered in Section \ref{sec-lex}, a solution for general $Q(u)$ can be reduced to that of quantifier-free $Q(u)$, subject to modification of the parameters $\comp{}{},\oplus,\prec$). 

In what follows, we make use of the fact that the quantifier-free formulas of $\WPAw$ are decidable, in the sense that whenever $A_0(x_1,\ldots,x_n)$ is quantifier-free with free variables $x_1,\ldots,x_n$ there is a closed term $t_A$ of System T so that $\HAw\vdash\forall x_1,\ldots,x_n(t_Ax_1\ldots x_n=1\leftrightarrow A_0(x_1,\ldots,x_n))$. This also means that the functional interpretation essentially interprets quantifier-free formulas as themselves.

Let us now fix closed terms $\comp{}{},\oplus,\prec$ and consider $\ZL_{\comp{}{},\oplus,\prec}$ as given in Definition \ref{def-paopen}, but where now $Q(u)$ is assumed to be quantifier-free. There are several variants of the negative translation which can be applied. Applying standard variant due to Kuroda, as used in \cite[Chapter 10]{Kohlenbach(2008.0)}, and using a few standard intuitionistic laws together with Markov's principle (all of which can be interpreted by the intuitionistic functional interpretation), it suffices to solve the functional interpretation of
\begin{equation}
\label{eqn-nzorn}
\begin{aligned}
\exists x^\sigma\forall d^\delta Q(\comp{x}{d})\to &\neg\neg\exists y^\sigma(\forall d\; Q(\comp{y}{d})\\
&\wedge \forall a\succ y\;\exists d\; \neg Q(\comp{y\oplus a}{d})).
\end{aligned}
\end{equation}
We must now apply the rules of Figure 1 to (\ref{eqn-nzorn}), which we do step by step. We first observe that the inner part of the conclusion of (\ref{eqn-nzorn}) within the double negations is translated to
\begin{equation*}
\exists y^\sigma, h^{\rho\to\delta}\forall d, a(Q(\comp{y}{d})\wedge (a\succ y\to \neg Q(\comp{y\oplus a}{ha})).
\end{equation*}
Therefore the double negated conclusion is partially interpreted (i.e. before the final instance of the $\forall$-rule) as
\begin{equation*}
\forall F,G\; \exists y,h(Q(\comp{y}{Fyh})\wedge (Gyh\succ y\to \neg Q(\comp{y\oplus Gyh}{h(Gyh)}))
\end{equation*}
where here $F:\sigma\to (\rho\to \delta)\to\delta$ and $G:\sigma\to (\rho\to\delta)\to \rho$. Therefore, interpreting the main implication, in order to solve the functional interpretation of (\ref{eqn-nzorn}) we must produce three terms $r,s,t$ which take as input $x,F,G$ and have output types $\delta,\sigma$ and $\rho\to\sigma$ respectively, and satisfy
\begin{equation}
\label{eqn-ndzorn}
Q(\comp{x}{r})\to Q(\comp{s}{Fst})\wedge (Gst\succ s\to \neg Q(\comp{s\oplus Gst}{t(Gst)}))
\end{equation}
where for readability we suppress the input parameters, so that $r$ should actually read $rxFG$ throughout, and similarly for $s$ and $t$. Though (\ref{eqn-ndzorn}) looks complicated, it can be given a fairly intuitive characterisation as follows. 

The original axiomatic formulation of Zorn's lemma is equivalent (using $\QFAC$) to the statement that given some $x^\sigma$ satisfying $\forall dQ(\comp{x}{d})$ we can find some $y$ also satisfying $\forall d Q(\comp{y}{d})$ together with an $h:\rho\to\delta$ witnessing maximality of $y$ in the sense that $\neg Q(\comp{y\oplus a}{ha})$ for any $a\succ x$. On the other hand, the computational interpretation of Zorn's lemma given as (\ref{eqn-ndzorn}) says that for any $x^\sigma$ together with `counterexample functionals' $F,G$ we can produce elements $s$ and $t$ (in terms of $x,F,G$), where $s$ \emph{approximates} our maximal element $y$ in the sense that it satisfies $Q(\comp{s}{d})$ not for all $d$ but just for $d:=Fst$, while $t$ approximates $h$ in the sense that it satisfies $\neg Q(\comp{s\oplus a}{t(a)}$ not for all $a\succ s$ but just for $a:=Gst$ whenever $Gst\succ s$. Indeed, this can be seen as a slightly more intricate version of Kreisel's no-counterexample interpretation, and the relationship between $\ZL_{\comp{}{},\oplus,\prec}$ and (\ref{eqn-ndzorn}) is similar to the relationship between Cauchy convergence and `metastability' (see \cite[Section 2.3]{Kohlenbach(2008.0)}).

\subsection{Solving the functional interpretation of $\ZL_{\comp{}{},\oplus,\prec}$}
\label{sec-dial-solve}

From this point onwards, we no longer need to deal directly with the functional interpretation. Rather, our focus is on solving the functional interpretation of $\ZL_{\comp{}{},\oplus,\prec}$ as given in (\ref{eqn-ndzorn}). To be more precise, we will construct realizing terms $r$, $s$ and $t$ which each take as input $x:\sigma$, $F:\sigma\to (\rho\to\delta)\to\delta$ and $G:\sigma\to (\rho\to\delta)\to\rho$ and satisfy,
\begin{equation*}
Q(\comp{x}{r})\to Q(\comp{s}{Fst})\wedge C(G,s,t)
\end{equation*}
for any input, where $C(G,y,h)$ abbreviates the formula
\begin{equation*}
C(G,y,h):\equiv Gyh\succ y\to \neg Q(\comp{y\oplus Gyh}{h(Gyh)}).
\end{equation*}
Interestingly, we do not require $\ZL_{\comp{}{},\oplus,\prec}$ in order to verify our realizing terms. Instead, we work in $\HAw$ extended with two recursively defined constants together with a simple universal axiom which we label `relevant part'. That this formal theory has a model is a separate question, which we discuss after presenting our main result (Theorem \ref{thm-mainver}).
\begin{definition}[$\HAw$]
Let $t_C$ denote the term of System T satisfying $t_CGyh=1\leftrightarrow C(G,y,h)$, which exists since $\prec$ is decidable and $Q(u)$ is quantifier-free.
\end{definition}
For the remainder of this section, we fix some closed term $e:(\sigma\to (\rho\to\delta)\to\delta)\to (\sigma\to (\rho\to\delta)\to\sigma)$ of System T, so that all definitions and results that follows are parametrised by $e$.
\begin{definition}[$\HAw$]
\label{def-self}
Define the new constant $\Omega_e:(\sigma\to (\rho\to\delta)\to\delta)\to\sigma\to\delta$ by
\begin{equation}
\label{eqn-omegabig}
\Omega_eFx=F\contd{x}{\Omega_e}{F}(\lambda a\; . \; \Omega_eF(\contd{x}{\Omega_e}{F}\oplus a)\mbox{ if $a\succ \contd{x}{\Omega_e}{F}$ else $0_\delta$})
\end{equation}
where $\contd{x}{\Omega_e}{F}$ is shorthand for
\begin{equation*}
\contd{x}{\Omega_e}{F}:=eFx(\lambda a\;. \; \Omega_eF(x\oplus a)\mbox{ if $a\succ x$ else $0_\delta$})
\end{equation*}
Furthermore, we use the abbreviation
\begin{equation*}
\Omega_{e,F,x}:=\lambda a\; . \; \Omega_eF(x\oplus a)\mbox{ if $a\succ x$ else $0_\delta$}
\end{equation*}
so that (\ref{eqn-omegabig}) can be expressed as $\Omega_e Fx=F\contd{x}{\Omega_e}{F}\Omega_{e,F,\contd{x}{\Omega_e}{F}}$ for $\contd{x}{\Omega_e}{F}=eFx\Omega_{e,F,x}$. 
\end{definition}

\begin{definition}[$\HAw+\Omega_e$]
We define in the language of $\HAw+\Omega_e$ the `relevant part' axiom for $\Omega_e$ as
\begin{equation*}
\SR_e\; : \; \forall x,F(\comp{x}{\Omega_eFx}=\comp{\contd{x}{\Omega_e}{F}}{\Omega_eFx}).
\end{equation*}
\end{definition}
Intuitively, the relevant part axiom says that if we take the approximation of $x$ of size $\Omega_eFx$, then actually this approximation has no more information than that of the truncated version $\contd{x}{\Omega_e}{F}$ of $x$, and so the latter already contains the `relevant part' of this approximation. We will see a natural example of an $e$ which satisfies this axiom in Section \ref{sec-lex}.

\begin{definition}[$\HAw+\Omega_e$]
Define the constant $\Gamma_e:(\sigma\to (\rho\to\delta)\to \delta)\to (\sigma\to(\rho\to\delta)\to\rho)\to\sigma\to\sigma^\ast$ in the language of $\HAw+\Omega_e$ by 
\begin{equation}
\label{eqn-gamma}
\Gamma_eFGx:=y::\begin{cases}\seq{} & \mbox{if $t_CGy\Omega_{e,F,y}=1$}\\ \Gamma_eFG(y\oplus Gy\Omega_{e,F,y}) & \mbox{otherwise}\end{cases}
\end{equation}
for $y:=\contd{x}{\Omega_e}{F}$, where here $y::l$ denotes the appending of $y$ to the front of the list $l$. i.e. $y::[l_1,\ldots,l_{j-1}]:=[y,l_1,\ldots,l_{j-1}]$.
\end{definition}
\begin{theorem}[$\HAw+\Omega_e+\Gamma_e+\SR_e$]
\label{thm-mainver}
Define terms $r$, $s$ and $t$ as follows:
\begin{equation*}
\begin{aligned}
rxFG&:=_\delta\Omega_{e}Fx\\
sxFG&:=_\sigma\tail(\Gamma_eFGx)\\
txFG&:=_{\rho\to\delta}\Omega_{e,F,\tail(\Gamma_eFGx)}
\end{aligned}
\end{equation*}
where $\tail(l)$ denotes the last element of the list $l$ (and $\tail([])=0_\sigma$). Then provably in $\HAw+\Omega_e+\Gamma_e+\SR_e$ we have
\begin{equation}
\label{eqn-goal}
\forall x,F,G(Q(\comp{x}{r})\to Q(\comp{s}{Fst})\wedge C(G,s,t))
\end{equation}
where in the above formula we write just $r$ instead of $rxFG$, and similarly for $s$ and $t$.
\end{theorem}

\begin{proof}
Fixing $F$ and $G$, we prove by induction on $n$ that
\begin{equation}
\label{eqn-ind}
\begin{aligned}
\forall x(|\Gamma_eFGx|=n&\wedge Q(\comp{x}{rx})\\
&\to Q(\comp{sx}{F(sx)(tx)})\wedge C(G,sx,tx))
\end{aligned}
\end{equation}
where here $rx$ is shorthand for $rxFG$ (i.e. the parameter $x$ is now explicitly written since it varies in the induction). Since $|\Gamma_eFGx|\geq 1$, our base case is $n=1$ which means that $t_CGy\Omega_{e,F,y}=1$ and $\Gamma_eFGx=\seq{y}$ for $y:=\contd{x}{\Omega_e}{F}$. But this implies that $sx=y$ and $tx=\Omega_{e,F,y}$, and thus in particular $C(G,sx,tx)$ holds. Next, we observe that
\begin{equation*}
\comp{x}{rx}=\comp{x}{\Omega_eFx}\stackrel{(a)}{=}\comp{sx}{\Omega_eFx}\stackrel{(b)}{=}\comp{sx}{F(sx)(tx)}
\end{equation*}
where (a) follows from $\SR_e$ and the definitions of $rx$ and $sx$, while for (b) we use that
\begin{equation*}
\Omega_eFx=Fy\Omega_{e,F,y}=F(sx)(tx).
\end{equation*}
Thus from $Q(\comp{x}{rx})$ we can infer $Q(\comp{sx}{F(sx)(tx)})$, which establishes (\ref{eqn-ind}) for $n=1$.

For the induction step, suppose that $|\Gamma_eFGx|=n+1$, which implies that $t_CGy\Omega_{e,F,y}=0$. Setting $y:=\contd{x}{\Omega_e}{F}$ as before, and in addition $a:=Gy\Omega_{e,F,y}$, by unwinding definitions it follows from $\neg C(G,y,\Omega_{e,F,y})$ that
\begin{enumerate}[(i)]

\item $a\succ y$ and thus $\Omega_{e,F,y}(a)=\Omega_eF(y\oplus a)$,

\item $Q(\comp{y\oplus a}{\Omega_{e,F,y}(a)})$ and thus $Q(\comp{y\oplus a}{\Omega_eF(y\oplus a)})$ by (i).

\end{enumerate}
Now since $\Gamma_eFGx=y::\Gamma_eFG(y\oplus a)$ and thus $|\Gamma_eFG(y\oplus a)|=n$, we can apply the induction hypothesis for $x':=y\oplus a$. Since $rx'=\Omega_eFx'=\Omega_eF(y\oplus a)$ it follows from (ii) that $Q(\comp{x'}{rx'})$ and therefore we have $Q(\comp{sx'}{F(sx')(tx')})$ and $C(G,sx',tx')$. But since
\begin{equation*}
sx=\tail(y::\Gamma_eFG(y\oplus a))=\tail(\Gamma_eFG(x'))=sx'
\end{equation*}
and similarly $tx=tx'$, it follows that $Q(\comp{sx}{F(sx)(tx)})\wedge C(G,sx,tx)$, which establishes (\ref{eqn-ind}) for $n'=n+1$. This completes the induction, and (\ref{eqn-goal}) follows by taking some arbitrary $F,G,x$ and letting $n:=|\Gamma_eFGx|$ in (\ref{eqn-ind}).
\end{proof}
The above result which solves the functional interpretation of $\ZL_{\comp{}{},\oplus,\prec}$ is valid for arbitrary $e$. However, it is only useful if the theory $\HAw+\Omega_e+\Gamma_e+\SR_e$ has a reasonable model. The final results of this section establish some conditions by which both $\Omega_e$ and $\Gamma_e$ give rise to total objects and hence exist in $\modcont$. An example of a setting where $\SR_e$ is also valid in $\modcont$ is given in Section \ref{sec-lex}.
\begin{theorem}[$\modpar$]
\label{thm-maintot}
Let $\Omega_e$ denote a fixed point of the primitive recursive defining equation (\ref{eqn-omegabig}) - where the closed primitive recursive term $e$ is interpreted as some total object in $\modpar$ - and suppose that there exist $<$ and $L$ such that $<$ is compatible with $(\oplus,\prec)$ and chain bounded w.r.t. $\comp{\cdot}{}$ and $L$. Suppose in addition that $\contd{\cdot}{\Omega_e}{F}\in P_{\sigma\to\sigma}$ and $\lambda x\; . \; Fx\Omega_{e,F,x}\in P_{\sigma\to\delta}$ form a truncation w.r.t. $\comp{\cdot}{}$, $L$ and $<$ for any total $F$. Then $\Omega_e$ is total.
\end{theorem}

\begin{proof}
This is a simple adaptation of Theorem \ref{thm-totcont}, taking $\Omega_e:=\lambda F\; . \; \Psi(eF)F$. If $F\in T_{\sigma\to (\rho\to\delta)\to\delta}$ then also $eF\in T_{\sigma\to (\rho\to\delta)\to\sigma}$ by totality of $e$, and thus whenever $\cont{\cdot}{\Psi}{eF,F}$ and $\lambda x\; . \; Fx\Psi_{eF,F,x}$ form a truncation w.r.t. $\comp{\cdot}{}$, $L$ and $<$ then $\Omega_eF=\Psi(eF)F\in T_{\sigma\to\delta}$. But the truncation condition is exactly that given as the statement of this theorem, and if this holds for arbitrary total $F$ then $\Omega_e$ is also total.
\end{proof}

\begin{theorem}[$\modpar$]
\label{thm-maintots}
Let $\Gamma_e$ denote a fixed point of the defining equation (\ref{eqn-gamma}). Under the assumptions of Theorem \ref{thm-maintot}, $\Gamma_e$ is total.
\end{theorem}

\begin{proof}
We can define $\Gamma_e:=\lambda F,G\; . \; \Psi(\omega F)(fFG)$ where $\omega$ and $f$ are total representations in $\modpar$ of the following functionals definable in $\HAw+\Omega_e$:
\begin{equation*}
\begin{aligned}
\omega Fxp&:=_\sigma \contd{x}{\Omega_e}{F}\\
fFGxp&:=_{\sigma^\ast}x::\begin{cases}[] & \mbox{if $t_CGx\Omega_{e,F,x}$}\\ p(Gx\Omega_{e,F,x}) & \mbox{otherwise}\end{cases}
\end{aligned}
\end{equation*}
where here $p:\rho\to\sigma^\ast$ (note that totality of $\omega$ and $f$ follows from totality of primitive recursive functionals plus totality of $\Omega_e$ as established in Theorem \ref{thm-maintot} above). To see that $\Gamma_e$ satisfies (\ref{eqn-gamma}) is just a case of unwinding the definitions.

Now, if $F$ and $G$ are total it follows that $\omega F$ and $fFG$ are also total, and so by Theorem \ref{thm-totcont}, $\Gamma_eFG=\Psi(\omega F)(fFG)$ is total if we can show that $\cont{\cdot}{\Psi}{\omega F,fFG}$ and $\lambda x.(fFG)x\Psi_{\omega F,fFG,x}$ form a truncation. But $\cont{x}{\Psi}{\omega F,fFG}=\contd{x}{\Omega_e}{F}$, and so this follows from the assumption that $\contd{\cdot}{\Omega_e}{F}$ and $\lambda x.Fx\Omega_{e,F,x}$ form a truncation. Formally, if $\contd{x}{\Omega_e}{F}$ is total for $x\in L$ (which it always is by totality of $\Omega_e$), then since in addition $Fy\Omega_{e,F,y}$ is total for $y:=\contd{x}{\Omega_e}{F}$ then $\contd{\cdot}{\Omega_e}{F}$ has a point of continuity $d$ for $x$. Condition (\ref{truncb}) follows trivially. Therefore we have shown that $\Gamma_e$ is total.
\end{proof}
\textbf{Remark.} Our use of controlled recursion means that there are no type level restrictions on the output types $\Omega_eFx:\delta$ or $\Gamma_eFGx:\sigma^\ast$. This not only permits a greater degree of generality but is essential even for simple applications: In Example \ref{ex-countableideal}, $\sigma:=\nat\to\bool$ and thus $\sigma^\ast$ is a higher type. 
%
\section{Application: The lexicographic ordering}
\label{sec-lex}

We conclude the paper by showing how our parametrised results can now be implemented in the special case of induction over the lexicographic ordering on sequences. This constitutes a direct counterpart to open induction as presented in \cite{Berger(2004.0)}, and is closely related to the recursive scheme introduced in \cite{Powell(2018.0)} for extracting a witness from the proof of Higman's lemma. 
\begin{definition}[$\HAw$]
\label{def-lex}
Let $\theta$ be some arbitrary type, and suppose that $\lhd:\theta\times\theta\to\bool$ is a decidable relation on $\theta$ such that induction over $\lhd$ is provable in $\HAw$. Setting $\sigma:=\nat\to\theta$, $\delta:=\nat$, $\rho:=\nat\times(\nat\to\theta)$ and $\nu:=\theta^\ast$, define
\begin{equation*}
\begin{aligned}
\comp{x}{n}&:=\seq{x(0),\ldots,x(n-1)}\\
x\oplus (n,y)&:=\comp{x}{n}\at y\\
(n,y)\succ x &:=y(n)\lhd x(n)
\end{aligned}
\end{equation*}
where $(\seq{x(0),\ldots,x(n-1)}\at y)(i):=x(i)$ if $i<n$ and $y(i)$ otherwise. We define $\LEX_{\rhd}$ to be the principle $\ZL_{\comp{\cdot}{},\oplus,\prec}$ for the parameters given above i.e.  
\begin{equation*}
\begin{aligned}
\exists x\forall d Q(\comp{x}{d})&\to \exists y(\forall d Q(\comp{y}{d})\\
&\wedge\forall (n,z)(z(n)\lhd y(n)\to \exists d\neg Q(\comp{\comp{y}{n}\at z}{d})).
\end{aligned}
\end{equation*}
\end{definition}
Our axiom $\LEX_\rhd$ is essentially the contrapositive of open induction as presented in \cite{Berger(2004.0)}, and as such the theory $\WPAw+\QFAC+\LEX_\rhd$ (for various instantiations of $\rhd$) is capable not only of formalizing large parts of mathematical analysis but also giving direct formalizations of minimal bad sequence arguments common in the theory of well quasi orderings. We now show how it can be given a direct computational interpretation using the theory developed so far.
\begin{lemma}[$\modpar$]
\label{lem-cclex}
Define $L\subset T_{\sigma}$ to be the set of all \emph{strict} total objects i.e. those satisfying $x(\bot)=\bot_\theta$ (recall that $\sigma=\nat\to\theta$), and let the partial order $<$ on $T_{\sigma}$ by defined by
\begin{equation*}
y>x:\Leftrightarrow \exists n\in\NN(\comp{y}{n}=_{\NN^\ast}\comp{x}{n}\wedge y(n)\lhd x(n))
\end{equation*}
where here $\lhd$ is now interpreted as a total functional $T_{\theta\times\theta\to \bool}$. In other words, $y>x$ if it is lexicographically \emph{smaller} than $x$ w.r.t. $\lhd$. Then $<$ is compatible with $(\oplus,\prec)$ and chain bounded w.r.t. $\comp{\cdot}{}$ and $L$.
\end{lemma}

\begin{proof}
Compatibility is clear, while chain boundedness follows easily using a standard construction for the lexicographic ordering. Take some nonempty chain $\gamma\subset T_\sigma$ and inductively define the sequence of total objects $u_k\in T_\theta$ for $k\in\NN$ by taking $u_k$ to be the $\lhd$-minimal element of the set
\begin{equation*}
S_k:=\{x(k)\; | \; x\in\gamma\mbox{ and $(\forall i<k)(x(i)=u_i)$}\}\subseteq T_\theta.
\end{equation*}
Note that $S_k$ are nonempty by induction on $k$, and $u_k$ is well-defined since the $\lhd$-minimum principle is provable from induction over $\lhd$, which is provable in $\HAw$ and thus satisfied by the total elements $T_\theta$. Now define $\ub{\gamma}(k):=u_k$ for $k\in\NN$ and $\ub{\gamma}(\bot)=\bot$, which is clearly an element of $L\subset T_{\nat\to\theta}$. It follows by definition that for any $d\in\NN$ there exists some $x\in\gamma$ with $\comp{x}{d}=\seq{u_0,\ldots,u_{d-1}}=\comp{\ub{\gamma}}{d}$. To see that $\ub{\gamma}$ is an upper bound, take some $x\in\gamma$ and assume that $x\neq \ub{\gamma}$. Let $k\in\NN$ be the least with $x(k)\neq \ub{\gamma}(k)=u_k$. Then by definition of $u_k$ there is some $y\in\gamma$ with $\comp{y}{k}=\seq{u_0,\ldots,u_{k-1}}=\comp{x}{k}$ and $y(k)=u_k$. Since $<$ is a total order on $\gamma$ we must have either $x<y$ or $y<x$, and since $x(k)\neq y(k)$ this means that either $x(k)\lhd y(k)$ or $y(k)\lhd x(k)$. But by minimality of $u_k=y(k)$ we must have $\ub{\gamma}(k)=y(k)\lhd x(k)$ and thus $\ub{\gamma}>x$. This proves that $x\leq \ub{\gamma}$ for any $x\in\gamma$.
\end{proof}

Our next step is to define a suitable closed term $e$ of $\HAw$ which not only induces a truncation in the sense of Theorem \ref{thm-maintot} but also satisfies $\SR_e$ in the total continuous functionals. For this, we introduce a powerful idea that is already implicit in Spector's fundamental bar recursive interpretation of the axiom of countable choice \cite{Spector(1962.0)}, and has been studied in more detail in \cite{OliPow(2012.2)}.

From now on we make the fairly harmless assumption that the canonical object $0_\theta$ is minimal w.r.t to $\lhd$ (this could in theory be circumvented but having it makes what follows slightly simpler). 
\begin{definition}[$\HAw$]
\label{def-spec}
For $x:\sigma$ and $n:\nat$ let 
\begin{equation*}
\exts{x}{n}:=\comp{x}{n}\at (\lambda i.0_\theta):\sigma,
\end{equation*}
and define the primitive recursive functional $\spec:(\sigma\to\nat)\to\sigma\to\sigma$ by
\begin{equation*}
\spec \phi xk:=_\theta \begin{cases}0_\theta & \mbox{if $(\exists i\leq k)(\phi(\exts{x}{i})<i)$}\\ x(k) & \mbox{if $(\forall i\leq k)(\phi(\exts{x}{i})\geq i)$}\end{cases}
\end{equation*}
where we note that the bounded quantifiers can be represented as bounded search terms in System T.
\end{definition}
\begin{lemma}[$\modpar$]
\label{lem-spec}
Let us represent $\spec$ in $\modpar$ by the total continuous functional 
%
%
\begin{equation*}
\eta \phi xk:=\begin{cases}0_\theta & \mbox{if $(\exists i\leq k)(\forall j\leq i(\phi(\exts{x}{j})\in\NN)\wedge \phi(\exts{x}{i})<i)$}\\ x(k) & \mbox{if $(\forall i\leq k)(\phi(\exts{x}{i})\in\NN\wedge\phi(\exts{x}{i})\geq i)$}\\ \bot & \mbox{otherwise}\end{cases}
\end{equation*}
with $\eta \phi x\bot=\bot$.\footnote{This is a standard domain theoretic interpretation of $\eta$, where the bounded search terminates with $0$ for the first $i\leq k$ it finds with $\phi(\exts{x}{i})<i$, and returns $\bot$ if $\phi(\exts{x}{i})$ is undefined for any $i$ that is queried.} Then for any $\phi\in P_{\sigma\to\nat}$, the functionals $\eta\phi\in P_{\sigma\to\sigma}$ and $\phi$ form a truncation w.r.t. $\comp{\cdot}{}$, $L$ and $<$.
\end{lemma}
\begin{proof}
Part (\ref{truncb}) is simple: Suppose that $x,y,\eta\phi x\in T_\sigma$ and $\eta\phi x<y$ so that there exists some $n\in\NN$ with $\comp{y}{n}=\comp{\eta\phi x}{n}$ and $y(n)\lhd \eta\phi x(n)$. Since we cannot have $y(n)\lhd 0_\theta$ by our minimality assumption, we must have $\eta\phi x(n)=x(n)$. But then by definition of $\eta$ it follows that $\eta\phi x(k)=x(k)$ for all $k<n$, and thus $\comp{y}{n}=\comp{x}{n}$ and so $x<y$.

For part (\ref{trunca}), let us now assume that $x\in L$ with $\eta\phi x\in T_\sigma$ and $\phi(\eta\phi x)\in\NN$. We first show that there exists some $n\in\NN$ with $\phi(\exts{x}{n})<n$. Suppose for contradiction that for all $i\in\NN$ we have either $\phi(\exts{x}{i})=\bot$ or $\phi(\exts{x}{i})\geq i$. The first possibility is ruled out since if $\phi(\exts{x}{i})=\bot$ then $\eta\phi xi=\bot$ contradicting totality of $\eta\phi x$. But this means that $\eta\phi x=x$ (since also $\eta\phi\bot=\bot=x(\bot)$). But then $\phi(\eta\phi x)=\phi x\in\NN$ and so by Lemma \ref{lem-seqcont} there exists some $d\in\NN$ such that $\phi x=\phi y$ whenever $x(i)=y(i)$ for all $i<d$. Now set $N:=\max\{\phi x+1,d\}$ and consider $y:=\exts{x}{N}$. Then $x(i)=y(i)$ for all $i<N$ and so also for all $i<d$, which implies that
\begin{equation*}
\phi(\exts{x}{N})=\phi x<\phi x+1\leq N
\end{equation*}
a contradiction. Therefore we have shown there exists some $n\in\NN$ with $\phi(\exts{x}{n})<n$, from which it follows that $\eta\phi x=\exts{x}{m}$ for the least such $m\in\NN$ with this property (again, $\phi(\exts{x}{j})\in\NN$ for all $j\leq m$ by totality of $\eta\phi x$). Let us now suppose that $y\in P_\sigma$ satisfies $\comp{x}{m}=\comp{y}{m}$. Then for $k<m$, since $\phi(\exts{y}{i})=\phi(\exts{x}{i})\geq i$ for all $i\leq k$ it follows that $\eta\phi yk=y(k)=x(k)$, and if $k\geq m$, since $\phi(\exts{y}{m})=\phi(\exts{x}{m})<m$ it follows that $\eta\phi yk=0$, and thus $\eta\phi y=\exts{x}{m}=\eta\phi x$.
\end{proof}
\begin{lemma}[$\modpar$]
\label{lem-lextot}
Let $\Omega_e$ be a fixed point of the equation (\ref{eqn-omegabig}) as in Theorem \ref{thm-maintot}, where now $e$ is defined by
\begin{equation*}
eFxp:=\eta(\lambda y.Fy(p|_y))x
\end{equation*}
for $\eta$ as in Definition \ref{def-spec} (resp. Lemma \ref{lem-spec}) and $$p|_y(n,z):=p(z)\mbox{ if $z(n)\lhd y(n)$ else 0}.$$ Then $\contd{\cdot}{\Omega_e}{F}\in P_{\sigma\to\sigma}$ and $\lambda x.Fx\Omega_{e,F,x}\in P_{\sigma\to\nat}$ form a truncation w.r.t. $\comp{\cdot}{}$, $L$ and $<$ for any $F$.
\end{lemma}
\begin{proof}
We first observe that
\begin{equation*}
\contd{x}{\Omega_e}{F}=eFx\Omega_{e,F,x}=\eta(\lambda y.Fy(\Omega_{e,F,x}|_y))x.
\end{equation*}
We now argue that for any $i\in\NN$ we have 
\begin{equation*}
\Omega_{e,F,x}|_{\exts{x}{i}}=\Omega_{e,F,\exts{x}{i}}.
\end{equation*}
For this we only need to check arguments $(n,y)$ which satisfy $(n,y)\succ \exts{x}{i}$ i.e. $y(n)\lhd (\exts{x}{i})(n)$. But by minimality of $0_\theta$ this is only possible if $n<i$ and $y(n)\lhd x(n)$, in which case 
\begin{equation*}
\begin{aligned}
\Omega_{e,F,x}|_{\exts{x}{i}}(n,y)&=\Omega_{e,F,x}(n,y)=\Omega_eF(\comp{x}{n}\at y)\\
&=\Omega_eF(\comp{\exts{x}{i}}{n}\at y)=\Omega_{e,F,\exts{x}{i}}(n,y).
\end{aligned}
\end{equation*}
Since $\eta\phi x$ only depends on $\phi$ for arguments of the form $\exts{x}{i}$, it follows that
\begin{equation*}
\contd{x}{\Omega_e}{F}=\eta\phi_{F,\Omega} x\mbox{ \ \ for \ \ }\phi_{F,\Omega}:=\lambda y.Fy\Omega_{e,F,y}.
\end{equation*}
But for any $F$, by Lemma \ref{lem-spec} applied to $\phi:=\phi_{F,\Omega}$ as defined above, we have that $\eta\phi_{F,\Omega}$ and $\phi_{F,\Omega}$ form a truncation w.r.t. $\comp{\cdot}{}$, $L$ and $<$, and the result follows.
\end{proof}
\begin{corollary}[$\modpar$]
\label{cor-modlex}
Let $\Omega_e$ and $\Gamma_e$ be fixed points of the equations (\ref{eqn-omegabig}) and (\ref{eqn-gamma}) respectively, for $e$ be as defined in Lemma \ref{lem-lextot}. Then $\Omega_e$ and $\Gamma_e$ are total, and thus $\modcont\models\HAw+\Omega_e+\Gamma_e$.
\end{corollary}

\begin{proof}
Directly from Lemmas \ref{lem-cclex} and \ref{lem-lextot} together with Theorems \ref{thm-maintot} and Theorem \ref{thm-maintots}.
\end{proof}
\begin{lemma}
\label{lem-modsr}
$\SR_e$ is valid in $\modcont$ for $e$ as in Lemma \ref{lem-lextot}.
\end{lemma}

\begin{proof}
The argument in the proof of Lemma \ref{lem-lextot} that $\contd{\cdot}{\Omega_e}{F}=\eta\phi_{F,\Omega}$ for $\phi_{F,\Omega}:=\lambda y.Fy\Omega_{e,F,y}$ is also valid in $\modcont$, and a simpler version of the argument in the proof of Lemma \ref{lem-spec} verifies that there is some $n\in\NN$ such that $\phi_{F,\Omega}(\exts{x}{n})<n$, and moreover $\contd{x}{\Omega_e}{F}=\eta\phi_{F,\Omega}x=\exts{x}{m}$ where $m\in\NN$ is the least satisfying this property. But since $\phi_{F,\Omega}(\exts{x}{m})=\phi_{F,\Omega}(\contd{x}{\Omega_e}{F})=\Omega_eFx$ and thus $\Omega_eFx<m$, it follows that
\begin{equation*}
\comp{x}{\Omega_eFx}=\comp{\exts{x}{m}}{\Omega_eFx}=\comp{\contd{x}{\Omega_e}{F}}{\Omega_eFx}
\end{equation*}
and so $\SR_e$ is satisfied.
\end{proof}

\begin{theorem}
For any type $\theta$ and relation on $\lhd$ such that induction over $\lhd$ is provable in $\HAw$, the functional interpretation of (the negative translation of) $\LEX_\lhd$ can be solved by a term in $\HAw+\Omega_e+\Gamma_e$, provably in $\HAw+\Omega_e+\Gamma_e+\SR_e$, for any closed term $e$ of System T. Moreover, defining $e$ as in Lemma \ref{lem-lextot}, we have $\modcont\models \HAw+\Omega_e+\Gamma_e+\SR_e$.
\end{theorem}

\begin{proof}
The first claim follows directly from Theorem \ref{thm-mainver}, and the second from Corollary \ref{cor-modlex} and Lemma \ref{lem-modsr}.
\end{proof}

\section{Conclusion and open questions}
\label{sec-conc}

In this paper, we explored various notions of recursion over chain bounded partial orders, and gave a general theorem on solving the functional interpretation of an axiomatic, parametrised form of Zorn's lemma.

We intend this work to be taken as the starting point for a number of much broader research questions in both proof theory and computability theory, which we hope to pursue in the future. These include the following:
\begin{enumerate}

\item Can particular instances of $\Phi$ and $\Psi$ as in Section \ref{sec-rec} be connected to known forms of strong recursion, particularly variants of bar recursion? We conjecture, for example, that $\Omega_e$ and $\Gamma_e$ as given in Section \ref{sec-lex} are definable using Spector's variant of bar recursion, using ideas from \cite{Powell(2014.0)}. Are more general results along the lines of \cite{BergOli(2006.0),EscOli(2015.0),OliPow(2012.2),Powell(2014.0)} possible?

\item The relationship between our simple and controlled recursors has many parallels to that between modified bar recursion and Spector's variant. It was shown in \cite{BergOli(2006.0)} that the former in fact defines the latter over System T. Under certain conditions, can we show that our simple recursor actually defines the controlled variant? It was also shown in \cite{BergOli(2006.0)} that Spector's bar recursion is S1-S9 computable in $\modcont$, but modified bar recursion is not. Does an analogous result hold in our setting?

\item Can we formulate Theorems \ref{thm-simpletot} and \ref{thm-totcont} so that they apply to \emph{non-continuous} models, such as the majorizable functionals \cite{Howard(1973.0)}?

\item What other applications of our abstract computational interpretation of Zorn's lemma are possible? Are there cases where a sensible choice of the parameters could lead to a more concise formalisation of a well-known proof, and consequently a more natural and efficient extracted program? In the other direction, can our framework be applied to give a computational interpretation to instances of Zorn's lemma stronger than even countable dependent choice?

\item If we were modify our formulation of Zorn's lemma so that chain boundedness is given as part of the \emph{syntactic} definition, rather than being implicitly dealt with in some model, how would we then solve its functional interpretation?

\end{enumerate}

\noindent\textbf{Acknowledgements.} The author is grateful to the anonymous referees for their valuable corrections and suggestions.

\end{document}